\documentclass[a4paper,11pt]{article}

\usepackage[headinclude]{typearea}
\usepackage[english]{babel}           
\usepackage[utf8]{inputenc}
\usepackage[T1]{fontenc}
\usepackage{scrtime}
\usepackage{amsmath,amsthm,amsfonts,amssymb}
\usepackage {color}                                  
\usepackage {pifont}                                 
\usepackage {multicol}                               
\usepackage {times}                                  
\usepackage[numbers]{natbib}                                                     
\usepackage{helvet}      
\usepackage {courier}                                
\usepackage {graphicx}                               
\usepackage {array}                                  
\usepackage {longtable}                              
\usepackage {makeidx}                                
\makeindex
\usepackage{fancyvrb}
\usepackage{enumerate} 
\usepackage{ifpdf}
\usepackage{caption}

\usepackage{setspace}

\usepackage{marvosym}

\theoremstyle{plain}
\newtheorem{theorem}{Theorem}[section]

\newtheorem{lemma}[theorem]{Lemma}
\newtheorem{proposition}[theorem]{Proposition}

\theoremstyle{definition}
\newtheorem{assumption}[theorem]{Assumption}
\newtheorem{definition}[theorem]{Definition}

\newtheorem{remark}[theorem]{Remark}

\newcommand{\E}{{\mathbb{E}}}

\renewcommand{\P}{{\mathbb{P}}}

\newcommand{\R}{{\mathbb{R}}}
\newcommand{\F}{\mathcal{F}}

\definecolor{darkgreen}{rgb}{0,0.5,0}

\newcommand{\tr}{\operatorname{tr}}

\renewcommand{\P}{{\mathbb P}}
\renewcommand{\F}{{\mathbb F}}

\newcommand{\1}{{\mathbf{1}}}

\newcommand{\cD}{{\cal D}}
\newcommand{\cE}{{\cal E}}
\newcommand{\cF}{{\cal F}}

\newcommand{\cL}{{\cal L}}

\newcommand{\cS}{{\cal S}}
\newcommand{\cU}{{\cal U}}

\newcommand{\be}{\begin{equation}}
\newcommand{\ee}{\end{equation}}
\newcommand{\bea}{\begin{eqnarray}}
\newcommand{\eea}{\end{eqnarray}}
\newcommand{\beast}{\begin{eqnarray*}}
\newcommand{\eeast}{\end{eqnarray*}}
\newcommand{\bproof}{\begin{proof}}
\newcommand{\eproof}{\end{proof}}

\newcommand{\domain}{{\R^d\times\R_0^+}}
\newcommand{\domnew}{{\R \times \R^{d-1}\times\R_0^+}}

\newcommand{\usol}{unique local solution}

\newcommand{\uss}{unique global strong solution}

\title{Optimal Control of an Energy Storage Facility Under a Changing Economic Environment and Partial Information}

\author{Anton A.~Shardin and Michaela Sz\"olgyenyi}

\begin{document}

\date{Preprint, April 2016}

\maketitle


\begin{abstract}

In this paper we consider an energy storage optimization problem in finite time in a model with partial information that allows for a changing economic environment.
The state process consists of the storage level controlled by the storage manager and the energy price process, which is a diffusion process the drift of which is assumed to be unobservable. We apply filtering theory to find an alternative state process which is adapted to our observation filtration. For this alternative state process we derive the associated Hamilton-Jacobi-Bellman equation and solve the optimization problem numerically. This results in a candidate for the optimal policy for which it is a-priori
not clear whether the controlled state process exists.
Hence, we prove an existence and uniqueness result for a class of time-inhomogeneous stochastic differential equations with discontinuous drift and singular diffusion coefficient.
Finally, we apply our result to prove admissibility of the candidate optimal control.\\

\end{abstract}


\centerline{\underline{\hspace*{18cm}}}

\noindent A.~A.~Shardin \\
Mathematical Institute, BTU Cottbus-Senftenberg, 03044 Cottbus, Germany\\

\noindent M.~Sz\"olgyenyi \Letter \\
Institute of Statistics and Mathematics, Vienna University of Economics and Business, 1020 Vienna, Austria\\
michaela.szoelgyenyi@wu.ac.at

\centerline{\underline{\hspace*{18cm}}}

\section{Introduction}
\label{sec:Introduction}

In this paper we study a stochastic optimization problem for energy storage management over a finite time horizon.
In order to incorporate unobservable changes of the economy into our model, we allow the energy price to depend on exogenous factors that cannot be observed directly.
Our model extends existing approaches from the literature in the following way: we allow for regime switching and partial information about the energy price process.
Assuming the agent has only incomplete information about the energy market makes the model more realistic and relevant for real application.

The aim of the agent is to maximize the profit from an energy storage facility by choosing the optimal rate of charging, or discharging the energy storage facility. This is achieved by buying at low- and selling at high prices.\\

\citet{thompson2009} point out the necessity of "{\em investment analysis and optimization methodologies capable of accurately accounting for the various operating characteristics of real storage facilities}" and mention 
"{\em enormous profit opportunities for storage operators}".
This is confirmed by \citet{ludkovski2010}: "{\em with the growing importance of energy commodities, sophisticated valuation of energy storage becomes an integral aspect of functioning financial markets.}"
"{\em Storage allows for inter-temporal transfer of the commodity and permits exploitation of the fluctuating market prices.}"\cite{ludkovski2010}\\

The motivation for solving the optimal control problem of an energy storage manager is two-fold.
The storage manager wants to act optimally in the market.
For this, she needs to know the optimal storage policy.
Hence she essentially aims to buy energy when prices are low, and sell energy when prices are high.
Furthermore, an agent in the energy market wants to know the value of the optimally controlled energy storage facility,
which corresponds to  the outcome of this optimization problem.\\

Stochastic optimization problems for energy storage facilities are studied by, e.g., \citet{chen2007}, \citet{ludkovski2010}, \citet{thompson2009}, and \citet{ware2013}. They consider natural gas storages and study the problem of determining the optimal charging/discharging strategy maximizing the expected discounted reward over a finite time horizon.

In contrast to that \citet{thompson2004} and \citet{zhao2009} study a stochastic optimal control problem for a pump-storage facility over a finite time horizon. They control the optimal amount of water that should be pumped/released in order to maximize the expected discounted reward from consuming/producing energy.
All these models assume a fixed economic environment as well as full information about the energy price process.
\citet{chen2010} extend \cite{chen2007} to allow for changes in modelling the drift. For this they assume the drift to be driven by an observable two-state Markov chain.\\

The energy price clearly depends on exogenous factors, which determine supply and demand,
where exogenous means, that these factors are independent of the driving source of uncertainty in the energy price process.
Permanent energy demand and supply shifts lead to structural changes in the energy price.

A reason for such changes can be seasonal differences.
For example, energy is consumed in winter for cooling and in summer for heating.
Furthermore, there are intraday fluctuations.
However, \citet{ahn2002} find that
"{\em storage value extracted from the seasonal spreads significantly underestimates the value obtained by optimal injection/withdrawal strategies.}"
Seasonalities cannot explain all of the energy price changes, which justifies to concentrate on additional determinants of price fluctuations.
Whereas seasonalities are clearly observable, other kinds of effects may not be.

The energy price depends on the economic environment, and on the geo-political situation.
These are the kind of effects we would like to focus on in our model. It is thus necessary to allow for a state dependence in the model for the energy price process.
Furthermore, as the energy market is liberalized, the energy price also depends on financial risks.
For this reason, we use a modelling approach from mathematical finance.\\

We model the energy price as a diffusion process the drift of which is driven by a Markov chain with non-observable state. There are economical as well as statistical reasons why it is more realistic to assume the underlying state process to be unobservable. A shift in the economic environment is often not trackable at the moment it happens, but only over time. Hence, we want to allow for uncertainty concerning the current state of the economy. From the statistical point of view the drift of a diffusion process is particularly hard to estimate over time, see, e.g., \citet[Chapter 4.2]{rogers2013}.\\

Stochastic optimization problems under partial information, especially those where the state process is influenced by an exogenous hidden Markov chain, have been studied intensively in the literature of mathematical finance and also in insurance mathematics, see, e.g., \citet{elliott1995}, \citet{lakner1995}, \citet{honda2003}, \citet{sass2004}, \citet{rieder2005},
\citet{frey2014}, \citet{sz12}, and \citet{sz16}.

\citet{shardin2016} model a pumped hydro storage in a changing economic environment, and assume partial information about the energy price process. However, for studying their optimization problem, they need to regularize their underlying process for technical reasons. Herein, we are able to spare out this step.\\

One of the outcomes of our model is that the candidate for the optimal control policy is of threshold type, i.e., there are certain levels of the energy price where the energy storage manager changes her behaviour. Thus, the controlled underlying process has a discontinuous drift coefficient. Additionally, due to the structure of the model, the diffusion coefficient is degenerate. Therefore, standard results about existence and uniqueness of solutions to stochastic differential equations (SDEs) do not apply. Hence, to prove admissibility of the resulting control policy, we need to prove an existence and uniqueness theorem extending the result of \citet{sz14} to the time-inhomogeneous case. A more detailed introduction to this issue is contained in Section \ref{sec:Existence}.\\

The main contribution of this paper is that we study a stochastic optimization problem under partial information in a model that allows for a changing economic environment for an energy storage facility and present an extensive numerical study of the outcomes of this model, namely the value function as well as the candidate for the optimal policy.
On the technical side we contribute that although our problem leads to a degenerate underlying process, all our results are obtained without any regularization techniques.
(For examples where regularization was applied, we refer to \cite{frey2014, shardin2016}.)
For verifying the candidate for the optimal strategy it is essential that the underlying controlled process exists.
We prove a theoretical result about existence and uniqueness to a class of time-inhomogeneous SDEs with discontinuous drift and degenerate diffusion coefficient, which generalizes the result from \cite{sz14}.\\

The paper is organized as follows.
In Section \ref{sec:opt} we formulate the energy storage optimization problem, describe the target function together with the underlying state process and derive the associated Hamilton-Jacobi-Bellman (HJB) equation.
Then we solve the optimization problem numerically in Section \ref{sec:Num}. The goal is to present an extensive study of the outcomes of our model.
In Section \ref{sec:Existence} we present an existence and uniqueness result for a class of time-inhomogeneous SDEs with discontinuous drift and degenerate diffusion coefficient. The (constructive) proof is moved to Appendix \ref{app:proof}.
Note that Section \ref{sec:Existence} and Appendix \ref{app:proof} are self contained and hence can be read isolated from Section \ref{sec:opt}.
In Section \ref{sec:Energy} we apply our result to prove admissibility of the candidate for the optimal control policy presented in Section \ref{sec:Num}, i.e., we show that the system of underlying state processes controlled by this policy has a \uss{}.
Section \ref{sec:Outlook} contains an outlook onto different fields of applied mathematics where the result on SDEs presented herein can potentially be applied.

\section{The energy storage optimization problem}
\label{sec:opt}

Let the filtered probability space $(\cE,\cF,\F,\P)$, where the filtration satisfies the usual conditions, carry all stochastic variables appearing herein. Further, let $(\widetilde B_t)_{t\ge0}$ be a standard Brownian motion.\\

We model the energy price process $S=(S_t)_{t\ge0}$ as an Ornstein-Uhlenbeck (OU) process,

\begin{equation}\label{eq:energy}
 dS_t = \kappa (\mu(Y_t)+K(t)-S_t) dt + \sigma d\widetilde B_t \, , \quad S_0=s_0 \,,
\end{equation}
with constant mean reversion speed $\kappa$, (seasonality free) mean reversion level $\mu:\R^D \longrightarrow \R$, seasonality function $K:[0,T] \longrightarrow \mathbb{R}$, and constant diffusion parameter $\sigma$. $\mu(Y_t)+K(t)$ is the long-term equilibrium price.

For the seasonality function we assume $K \in {C}^2([0,T])$, and satisfies a linear growth condition, i.e., there exist constants $c_1,c_2>0$ such that $|K(t)| \le c_1+c_2|t|$ for all $t \in [0,T]$. An example for the choice of this cyclical component is
\[
K(t)=K_S \cos \Bigl( \frac{2 \pi (t-t_S)}{\Delta} \Bigr) \, , 
\]
where $K_S$ determines the seasonal price trend, $t_S$ represents the time of the seasonal peak of the equilibrium price, and $\Delta$ is the length of one season, c.f.~\cite[Section 2.2]{chen2008b} and \cite[Section 3]{thompson2004}.

$Y=(Y_t)_{t\ge0}$ is an exogenous Markov chain with a $D$-dimensional state space, which is w.l.o.g.~given by $E_D=\{e_1,\ldots,e_D\}$, where $e_i$ is the $i$-th unit vector in the $\R^D$, $i=1,\ldots,D$. We denote the given intensity matrix of $Y$ as $\Lambda=(\lambda_{ij})_{i,j=1}^D$ and assume knowledge of the initial distribution of $Y$. We assume $\mu(e_i)=\mu_i$, $i=1,\ldots,D$, i.e., $\mu$ is constant for each state. W.l.o.g.~we assume $\mu_1>\ldots>\mu_D$.

The energy price is observable on the spot market, while the mean-reversion level is hard to estimate. Therefore, we assume that the information available to the energy storage manager is contained in the augmented filtration generated by the energy price process $\mathbb{F}^{S}=(\mathcal{F}_t^{S})_{t \ge 0}$, where $\mathcal{F}_t^{S}=\sigma\{ S_r , r \le t \}$.

We choose an OU process to model the energy price, since on the energy market it is necessary to allow for negative prices, and since mean-reversion is a stylized fact for energy prices.\\

Based on this energy price the manager of an energy storage facility has to decide when to charge, or discharge the storage and at which rate. Therefore, the rate of charging $u=(u_t)_{t\ge0}$ serves as the control. For discharging the quantity is negative. We model the energy storage level $Q=(Q_t)_{t\ge0}$ as
\begin{equation}
\label{fuellstandsde}
dQ_t=u_t \,dt \, , \quad Q_0=q_0 \, . 
\end{equation}

Since $\mu$ is not directly observable, we are in a situation of partial information. To overcome this issue, we apply filtering theory.

\paragraph{Filtering problem.} Since $\mu$ is not adapted to the observation filtration $\F^{S}$, it follows that the whole drift of $S$ is not adapted to $\F^{S}$. For convenience we denote this drift as
$$
a_t:=a(S_t,Y_t,t):=\kappa (\mu(Y_t)+K(t)-S_t)=\sum\limits_{i=1}^{D} \kappa (\mu_i+K(t)-S_t) \1_{\{  Y_t = e_i   \}} =:\sum_{i=1}^D a(S_t,e_i,t) \1_{\{Y_t=e_i\}}\, .
$$ 
Let the initial distribution of $Y$ be denoted as $\nu_0=(\nu_0^1,\ldots,\nu_0^D)$ and let $\nu_0^i>0$ for all $i=1,\ldots,D$.
Now, we estimate the drift $(a(S_t,Y_t,t))_{t \ge 0}$ given the information generated by the energy price process $S$ following \cite{liptser1977, wonham1964}:
\begin{equation}
\begin{split}
\label{wonhamfilterdrift}
\widehat{a}_t&:=\mathbb{E}\Bigl ( a(S_t,Y_t,t)  \,  \Bigl | \, \mathcal{F}_t^{S} \Bigr)=\mathbb{E}\Bigl ( \sum\limits_{i=1}^{D} a(S_t,e_i,t) \mathbf{1}_{\{  Y_t = e_i   \}} \,  \Bigl |  \, \mathcal{F}_t^{S} \Bigr)\\
&=\sum\limits_{i=1}^{D} a(S_t,e_i,t) \,\mathbb{P}\Bigl(    Y_t = e_i   \, \Bigl | \, \mathcal{F}_t^{S} \Bigr)=\sum\limits_{i=1}^{D} a(S_t,e_i,t) \pi_t^i    \, ,
\end{split}
\end{equation}
where $\pi_t^i$, $i=1,\ldots,D$, $t\ge0$, is the conditional probability for $Y$ being in state $i$ at time $t$. We further denote $\pi=(\pi_t)_{t\ge 0}$ and $\pi_t=(\pi_t^1,\ldots,\pi_t^D)$.
The following proposition can be found in \citet[Theorem 9.1]{liptser1977}.

\begin{proposition}\label{prop:filter}
The conditional probabilities
$$
\pi_t^i:=\mathbb{P}(Y_t=e_i \mid \mathcal{F}_t^{S}) \, , \quad e_i \in E_D \, , \quad i=1,\ldots,D\,,
$$ 
solve the following system of SDEs:
\begin{equation}
\label{filtergleichung1}
\pi_t^i=\nu_0^i+\int\limits_0^t \Bigl(\sum\limits_{k=1}^{D} \lambda_{ki} \pi_r^k\Bigr)  dr +\int\limits_0^t \pi_r^i \frac{\bigl( a(S_r,e_i,r)-\widehat{a}_r   \bigr)}{\sigma}dB_r \, , \quad i=1,\ldots,D \, ,
\end{equation}
where $B=(B_t)_{t \ge 0}$ is an $\mathbb{F}^{S}$-adapted Brownian motion given by
$$
B_t=\int\limits_0^t \frac{dS_r-\widehat{a}_r dr}{\sigma}\,,
$$
often referred to as innovation process.
\end{proposition}

In particular this means that $\widehat{a}$ is Markovian, and
\begin{equation}
\label{filtermarkov}
\widehat{a}_t =\widehat{a}(S_t,\pi_t,t)=\kappa \Bigl( \bigl< \mu,\pi_t \bigr>+K(t)-S_t\Bigr) \, ,  \quad \bigl< \mu,\pi_t \bigr>:=\sum\limits_{i=1}^D \mu_i \pi_t^i\,.
\end{equation}
Finally, we find an equivalent representation for the energy price process \eqref{eq:energy} in terms of the innovation process $B$:
\begin{equation*}
dS_t=\widehat{a}(S_t,\pi_t,t) \, dt +\sigma dB_t \, , \quad S_0=s_0  \, ,
\end{equation*}
where $\widehat{a}$ is defined in \eqref{filtermarkov}. The advantage in this representation is that now all ingredients of our model are adapted to the observation filtration and hence the optimization problem is well-defined.\\

Thus, we have transformed the original model under partial information to the following model under full information:
\begin{equation}
\begin{alignedat}{3}\label{eq:system}
dS_t&=\widehat{a}(S_t,\pi_t,t) \, dt +\sigma dB_t \, , &S_0&=s_0  \, ,\\
dQ_t&=u_t \,dt \, ,  &Q_0&=q_0  \, ,\\
d\pi_t&=\Lambda^{\top} \pi_t\, dt+ \sigma^{-1}\Bigl(\mathrm{diag}(\pi_t) a_t-\widehat{a}_t \pi_t \Bigr) d B_t \, , \quad &\pi_0&=\nu_0\,.
\end{alignedat}
\end{equation}
If $S,Q,\pi$ are started at time $t\in[0,T]$, we denote the starting values by $s,q,\nu=(\nu_1,\ldots,\nu_D)$, respectively.

In the end we remark that we can reduce the dimension of \eqref{eq:system} by one by setting $\pi_t^D=1-\sum_{i=1}^{D-1}\pi_t^i$ for all $t\ge0$.
Therefore, the correct state space for the filter is the simplex $\cS:= \{(\nu_1,\dots,\nu_{D-1}) \in (0,1)^{D-1}: \sum_{i=1}^{D-1} \nu_i < 1\}$ with closure $\bar \cS:= \{(\nu_1,\dots,\nu_{D-1}) \in [0,1]^{D-1}: \sum_{i=1}^{D-1} \nu_i \le 1\}$.
Hence the cost for obtaining full information are $D-1$ additional dimensions for the filter.

\paragraph{The optimization problem.} We introduce lower and upper capacity bounds of the energy storage facility $0 \le \underline{q} \le Q_t \le \overline{q}$ and say that the energy storage is full, or empty at $\overline q$, or $\underline q$, respectively. For the control $u$ we assume $u \in {\cU}$, where ${\cU}$ is the set of admissible controls imposing that $u$ is progressively measurable 
 and of feedback type (cf. \cite[Section 2.1]{ludkovski2010} and \cite[Section 2]{ware2013}).
Furthermore, the process $(S,Q,\pi)$ controlled by $u$ needs to exist, i.e., system \eqref{eq:system} needs to have a solution.
Additionally we impose conditions on the control such that the energy storage facility is modelled properly from an operational point of view.
Concretely, charging and discharging should not be possible at the maximum rate, if the energy storage facility is nearly full or empty, respectively.
For technical reasons these bounds are modelled in a continuous way, i.e., we impose smooth transition (cf.~\cite[Section 2.1]{chen2007}, \cite[Section 3]{thompson2009}, and \cite[Section 4.1]{ware2013}).
So we have
$u_t \in [\underline{u}(Q_t),\overline{u}(Q_t)]$
for $t \in [0,T]$, where $\underline u, \overline u: \R \longrightarrow \R$, with  $-M_u \le \underline{u}(q) \le 0$, $0 \le \overline{u}(q) \le M_u$, and $\underline u (\underline q)=\overline u(\overline q)=0$ are sufficiently smooth and bounded functions describing the maximal rates at which energy is sold or bought, respectively, and $M_u$ is constant.
Figure \ref{fig:admcontrols} illustrates the functions $\underline{u},\overline{u}$ on $[\underline q,\overline q]$. Note that for certain energy storage facilities, e.g., gas storages, $\underline{u},\overline{u}$ are falling for all levels of $q$. However, as we study a general energy storage facility, we do not lay down on a specific physical law determining these bounds. The interested reader may consult, e.g., \cite{thompson2009}.\\

\begin{figure}[h]
\begin{center}
\includegraphics[scale=0.3]{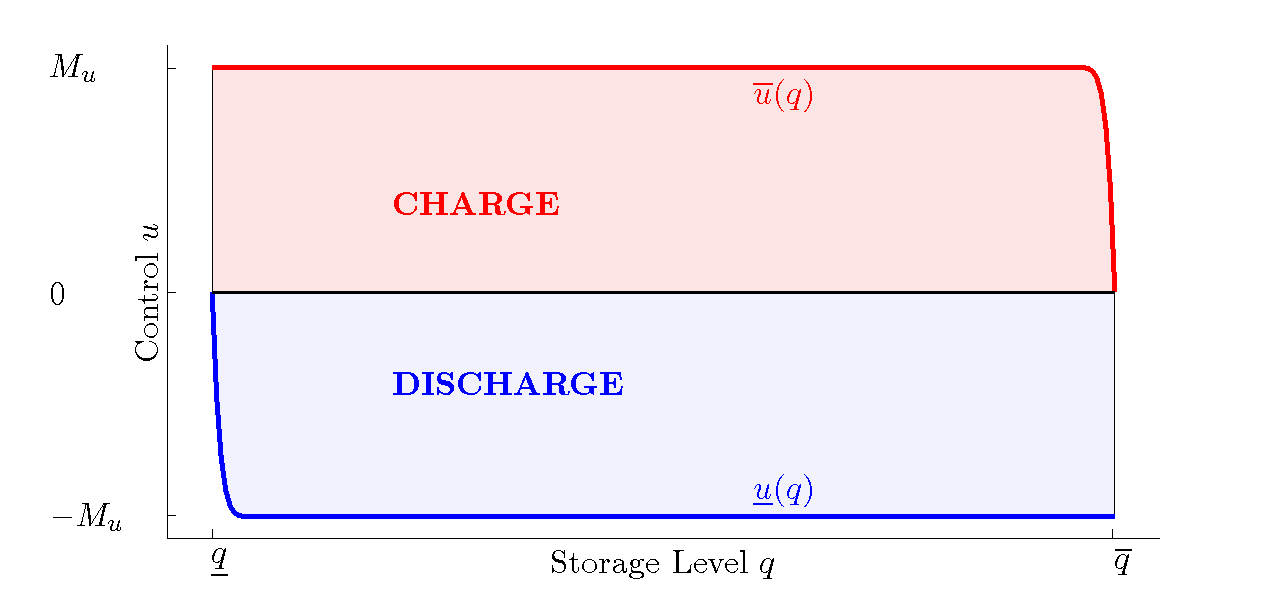}
\end{center}
\caption{The functions $\underline u$ and $\overline u$.}\label{fig:admcontrols}
\end{figure}

Our choice of $\underline u, \overline u$ guarantees that the process $Q$ is reflected at the boundary and stays inside $[\underline q,\overline q]$, if started inside.
Note that our model includes the most important operational conditions from a physical point of view, but neglects detailed modelling of some technical requirements, as the main purpose of this paper is to gain an impression of the effect of unobservable exogenous impacts.\\

The aim of the energy storage manager is to maximize the profit from the energy storage facility.
Define ${F}: \mathbb{R}  \times [\underline{q},\overline{q}] \times [-M_u,M_u] \longrightarrow \mathbb{R}$ by
\begin{eqnarray*}
\label{objective}
{F}(s,q,u) &=& \left\{
\begin{array}{lll}
-u \, \Bigl(s+d_+  \Bigr ) -c_0 \,q, & u \ge 0 & \text{(charging)} \\[1ex]
-u \, \Bigl(s-d_-  \Bigr)-c_0 \,q, & u < 0 & \text{(discharging)} \,,
\end{array}
\right.
\end{eqnarray*}
with fixed costs $d_{\pm}\ge 0$ and with storage costs $c_{0} \ge 0$. For the definition of $F$, cf. \cite[Section 2.3]{chen2007}, \cite[Section 2]{ludkovski2010}, \cite[Section 2]{thompson2009}, and \cite[Section 2]{ware2013}.\\

At the end of the planning horizon $T$, the value of the energy remaining in the storage facility is the reward that would be gained on immediately selling everything on the market at price $c_S S_T$, $0<c_S<1$. $c_S$ is a scrap rate that has to be paid by the storage manager.
Hence we define the {\em terminal reward function} $\Phi: \mathbb{R} \times [\underline{q},\overline{q}] \longrightarrow \mathbb{R}$ by 
\begin{equation}
\label{terminalfunktion}
\Phi(S_T,Q_T)=Q_T \Bigl(c_S S_T-d_-\Bigr) \,.
\end{equation}

Now we are ready to formulate the stochastic optimization problem under study. Denote $\Omega:=\mathbb{R} \times [\underline{q},\overline{q}] \times \bar \cS \times [0,T]$. We would like to maximize the expected discounted profit $J:\Omega \longrightarrow \mathbb{R}$ given by
\begin{equation*}
\label{rewardfunction2inf}
J(s,q,\nu,t;u)=\mathbb{E}_{s,q,\nu,t} \Bigl [ \int\limits_t^{T} e^{-\rho (r-t) }{F}(S_r,Q_r,u_r)dr + e^{-\rho(T-t)}\Phi(S_T,Q_T)\Bigr ]
\end{equation*}
by finding the optimal control $u \in \cU$, where $\rho>0$ is a constant discount rate and $\E_{s,q,\nu,t}(\cdot)$ is the expectation with the starting values $S_t=s,Q_t=q,\pi_t=\nu$.

The aim is to find the optimal value function $V:\Omega\longrightarrow \mathbb{R}$ given by
\begin{equation}
\label{wertfunktion2inf}
V(s,q,\nu,t)=\sup_{u \in \cU} J(s,q,\nu,t;u) \, ,
\end{equation}
and (if it exists) the optimal control policy $u^*=\mbox{argmax}_{u \in \cU} J(s,q,\nu,t;u)$.\\

\begin{lemma}
 The value function $V$ is continuous and satisfies a linear growth condition, i.e., there are constants $c_1,c_2>0$ such that $|V(s,q,\nu,t)|\le c_1+c_2\|(s,q,\nu,t)\|$.
\end{lemma}

\begin{proof}
 Follows directly from \cite[Chapter 3, Theorem 5]{krylov1980} and from the construction of $F$ and $\Phi$.
\end{proof}


\begin{remark}
The value function $V$ is also associated as the fair price of the real option, which represents the value of the energy storage facility at time $t$, if sold at time $T$.
\end{remark}

Applying the dynamic programming principle (cf. \cite[Section 3.3, Theorem 3.3.1]{pham2009}), we derive the HJB equation associated with system \eqref{eq:system} and optimization problem \eqref{wertfunktion2inf},
\begin{equation}
\begin{aligned}
\label{hjbTinf}
V_t+(\mathcal{L}-\rho)V+\sup_{{u} \in [\underline{u}(q),\overline{u}(q)]}\Bigl(u \, V_{q}+ {F}\Bigr ) &= 0 \, ,\\
V(s,q,\nu,T)&=\Phi(s,q)
\end{aligned}
\end{equation}
for all $(s,q,\nu,t) \in \Omega$
where 
\begin{align*}
\mathcal{L}{V}(s,q,\nu,t) &= \kappa \Bigl(\sum\limits_{i=1}^D \mu_i \pi_t^i+K(t) -s\Bigr) V_{s}(s,q,\nu,t) + \frac{1}{2}{\sigma}^2 V_{ss}(s,q,\nu,t) \\
&+\sum_{i=1}^D \Bigl( \sum_{k=1}^D(\lambda_{ki} \pi_t^k)\Bigr) V_{\nu_i}(s,q,\nu,t)+ \frac{1}{2}\sum\limits_{i=1}^D \kappa \pi_t^i \Bigl( \mu_i-\sum\limits_{k=1}^D \mu_k\pi_t^k  \Bigr) V_{s \nu_i}(s,q,\nu,t) \\
&
 + \frac{1}{2}\sum_{i,j=1}^D \frac{\kappa^2}{\sigma^2} \pi_t^i \pi_t^j \Bigl(\mu_i-\sum\limits_{k=1}^D \mu_k\pi_t^k\Bigr) \Bigl(\mu_j-\sum\limits_{k=1}^D \mu_k\pi_t^k\Bigr) V_{\nu_i\nu_j}(s,q,\nu,t)  \, .
\end{align*}

Solving the pointwise optimization problem results in the following candidate for the optimal charging policy:
\begin{equation*}
u_t^*=u^*(s,q,\nu,t)=
\begin{cases}
\overline{u}(q)\,,  & s \le \overline{s}(s,q,\nu,t) \\
0\,, & \overline{s}(s,q,\nu,t) < s < \underline{s}(s,q,\nu,t)   \\
\underline{u}(q)\,,  & \underline{s}(s,q,\nu,t) \le s  \, ,\\
\end{cases}
\end{equation*}
where
\begin{align}\label{eq:sws}
\overline{s}(s,q,\nu,t)=({V}_{q}(s,q,\nu,t)-d_{+})
\quad \text{and} \quad
\underline{s}(s,q,\nu,t)=({V}_{q}(s,q,\nu,t) + d_{-})
\end{align}
are the boundaries at which the strategy switches between buying energy, waiting, and selling energy.\\

As the HJB equation is degenerate parabolic, we have little hope to find an analytic solution, or to show existence of a classical solution.
For a problem under full information \citet{chen2007} bring up the possibility to prove verification in terms of viscosity solutions, which is a concept of weak solutions that only requires continuity of the solution, but is still strong enough to prove uniqueness and furthermore, to build the basis of numerical experiments.
For the concept of viscosity solutions we refer the interested reader to \citet{crandall1992}, \citet{fleming2006}, and to \citet{barles1991} for numerics.\\

Let $\partial \bar \Omega,\partial \bar \cS$ denote the boundary of the closure of $\Omega, \cS$, respectively.
We start with a definition of viscosity solutions to \eqref{hjbTinf}.

\begin{definition}[viscosity solution]\label{def:viscosity}

\begin{enumerate}
 \item A function $w:\bar{\Omega} \longrightarrow \R$ is a {\it viscosity subsolution} to \eqref{hjbTinf}, if
 \[
  \phi_t(\bar{s},\bar q,\bar{\nu},\bar t)+(\mathcal{L}-\rho)\phi(\bar{s},\bar q,\bar{\nu},\bar t)+\sup_{{u} \in [\underline{u}(q),\overline{u}(q)]}\Bigl(u \, \phi_{q}(\bar{s},\bar q,\bar{\nu},\bar t)+ F(\bar{s},\bar q,u)\Bigr ) \ge 0
 \]
 for all $(\bar{s},\bar q,\bar{\nu},\bar t) \in \Omega$ and for all $\phi \in C^2(\Omega)$ such that $w-\phi$ attains a maximum at $(\bar{s},\bar q,\bar{\nu},\bar t)$ with $w(\bar{s},\bar q,\bar{\nu},\bar t)=\phi(\bar{s},\bar q,\bar{\nu},\bar t)$.
 
 \item A function $w:\bar{\Omega} \longrightarrow \R$ is a {\it viscosity supersolution} to \eqref{hjbTinf}, if
 \[
  \psi_t(\bar{s},\bar q,\bar{\nu},\bar t)+(\mathcal{L}-\rho)\psi(\bar{s},\bar q,\bar{\nu},\bar t)+\sup_{{u} \in [\underline{u}(q),\overline{u}(q)]}\Bigl(u \, \psi_{q}(\bar{s},\bar q,\bar{\nu},\bar t)+ F(\bar{s},\bar q,u)\Bigr ) \le 0
 \]
 for all $(\bar{s},\bar q,\bar{\nu},\bar t) \in \Omega$ and for all $\psi \in C^2(\Omega)$ such that $w-\psi$ attains a minimum at $(\bar{s},\bar q,\bar{\nu},\bar t)$ with $w(\bar{s},\bar q,\bar{\nu},\bar t)=\psi(\bar{s},\bar q,\bar{\nu},\bar t)$.
 
\item $w:\bar{\Omega} \longrightarrow \R$ is a {\it viscosity solution} to \eqref{hjbTinf},  if it is both a viscosity sub- and supersolution.
\end{enumerate}
\end{definition}

From \cite[Theorem 4.3.1]{pham2009} we know that the value function $V$ is a viscosity solution of \eqref{hjbTinf}.
The standard argument for proving uniqueness is to conclude from the fact that two viscosity solutions are equal on the boundary of the domain, that they are also equal in the interior.
However, we neither have boundary conditions in the direction of $q$, nor in the direction of $\nu$.
Since $Q$ can never leave $[\underline q,\overline q]$, if started inside, we can as well define $Q$ on the whole $\R$.
(In practice we will always start with a value inside the capacity bounds of the energy storage facility.)
For this reason and since $Q$ is decoupled from $S,\pi$, we do not require boundary conditions in the direction of $q$, cf. \cite[p.~429 and p.~431]{ware2013}.

Further on, from \cite[Corollary 2.2]{chigansky2007} we know that the HMM filter reaches the boundary $\partial \bar \cS$ with probability zero in finite time. Hence we can argue in terms of \cite[Theorem II.2 and Theorem II.3 together with Remark II.4]{lions1983b} to get uniqueness of the viscosity solution to \eqref{hjbTinf}.

It remains to prove that if \eqref{hjbTinf} has a smooth solution, then this is the optimal value function.

\begin{proposition}\label{prop:verification}
Let $v$ be a viscosity supersolution of \eqref{hjbTinf}, and $v \in C^2$ almost everywhere. Then $V \le v$.\\
Furthermore, if there is a strategy $\tilde{u} \in \cU$ such that $J(\tilde{u})$ is a viscosity supersolution with $J(\tilde{u}) \in C^2$ almost everywhere, then $J(\tilde{u})=V$, and $\tilde{u}$ is the optimal charging policy.
\end{proposition}

 \begin{proof}
Let $\varphi(s,q,\nu,t) := \frac{1}{\pi^\frac{D+2}{2}} e^{-\left( s^2 + q^2 + \sum_{i=1}^{D-1} \nu_i^2 +r^2\right) }$, where $\pi$ is the area of a circle with radius $1$ and
\[
 \varphi^n(s,q,\nu,t) := n^{D+2} \int_{-\infty}^\infty \dots \int_{- \infty}^\infty v(s-r_1,q-r_2,(\nu_1-r_3,\dots \nu_{D-1}-r_{D+1}),t-r_{D+2}) \varphi(nr)\, dr_1\, \dots \, dr_{D+2}\,,
\]
where $nr=(nr_1,\dots,nr_{D+2})$.
For $n \to \infty$, $\varphi^n \to v$ and $\cL \varphi^n \to \cL v$, see \cite{wheeden1977}.\\

Let $u=(u_t)_{t \ge 0}$ be an admissible strategy. Then
 \begin{align*}
e^{-\rho (T -t)} \varphi^n (S_T,Q_T,\pi_T,T)
=& \varphi^n(s,q,\nu,t) + \int_t^{T}  e^{-\rho (r-t)} \left[ -\rho \varphi^n(S_r,Q_r,\pi_r,r) + \varphi^n_r(S_r,Q_r,\pi_r,r) \right.\\ &\left.+\cL \varphi^n (S_r,Q_r,\pi_r,r) + u_r \varphi_x^n (S_r,Q_r,\pi_r,r) \right]\,dr
+ M_t\,,
\end{align*}

where $(M_t)_{t\ge 0}$ is a martingale.
Since $v$ is a viscosity supersolution which is of class $C^2$ a.e., it satisfies
\[
 -\rho v + v_t+\cL v + F(u)+u\, v_q \le 0 \quad \mbox{a.e.}\,,
\]
which allows us to choose $n$ sufficiently large such that
\[
 -\rho \varphi^n+\varphi^n_t + \cL \varphi^n + F(u)+u\, \varphi^n_q \le \varepsilon\,.
\]
Taking expectations we get
\begin{align*}
 &\E_{s,q,\nu,t} \left(e^{-\rho (T -t)} \varphi^n (S_T,Q_T,\pi_T,T)\right)
= \varphi^n(s,q,\nu,t) + \E_{s,q,\nu,t} \left(\int_t^{T}  e^{-\rho (r-t)} \left[ -\rho \varphi^n(S_r,Q_r,\pi_r,r)\right.\right.\\& \left.\left. +\varphi^n_r(S_r,Q_r,\pi_r,r)+ \rho \varphi^n(S_r,Q_r,\pi_r,r) - \varphi^n_r(S_r,Q_r,\pi_r,r)-F(S_r,Q_r,u_r)-u_r \varphi^n_q + \varepsilon + u_r \varphi_x^n (S_r,Q_r,\pi_r,r) \right]\right)\,.
\end{align*}

As $n\longrightarrow\infty$ we obtain
\begin{align*}
\E_{s,q,\nu,t} \left( e^{-\rho (T-t)} v (S_T,Q_T,\pi_T,T)\right)
\le v(s,q,\nu,t) - \E_{s,q,\nu,t} \left( \int_t^{T}  e^{-\rho (r-t)} F(S_r,Q_r,u_r) \,dr
\right)
\end{align*}
by dominated convergence and as $\varepsilon>0$ was arbitrary.
Thus,
\begin{align*}
&\E_{s,q,\nu,t} \left( e^{-\rho (T-t)} v (S_T,Q_T,\pi_T,T)\right)+\E_{s,q,\nu,t} \left( \int_t^{T}  e^{-\rho (r-t)} F(S_r,Q_r,u_r) \,dr\right)\\
&=\E_{s,q,\nu,t} \left( \int_t^{T}  e^{-\rho (r-t)} F(S_r,Q_r,u_r)\,dr+e^{-\rho (T-t)} \Phi (S_T,Q_T) \right)
=J(s,q,\nu,t;u)
\le v(s,q,\nu,t)\,.
\end{align*}
By taking the supremum over all $u \in [\underline u(q),\overline u(q)]$ in the derivation we get
$V(s,q,\nu,t) \le v(s,q,\nu,t)$.
\end{proof}


\section{Numerical study}
\label{sec:Num}

In this section we solve the HJB equation \eqref{hjbTinf} numerically.

The focus of this paper lies on unobservable factors influencing the energy price.
As seasonalities are clearly observable, these effects can be identified and we
may assume that energy prices in our model are already made free from them.
Hence, for the numerical part -- in order to be able to identify the effects from unobservable factors -- we set $K(t)\equiv 0$.

We apply the semi-Lagrangian approach with fully implicit timestepping, see \cite[Section 3]{chen2007}, together with a finite difference method, see \cite[Chapter 3, Section 5.3]{samarskii2001}, to obtain approximations to the value function as well as the optimal charging policy.
The method converges to the viscosity solution of the HJB equation, since the conditions from \cite[Section 4]{chen2007} are satisfied.

Our numerical study is conducted for a Markov chain with two states, i.e., $D=2$ and for the following parameter set:
$\mu_1=50$, $\mu_2=30$, $\kappa=15$, $\sigma=50$, $\lambda_{11}=\lambda_{22}=-0.5$, $\rho=0.05$, $T=1$ year, $c_0=0$, $d_{+}=d_{-}=10$, $c_S=0.95$, $\underline q=0$, $\overline q=100$, and $M_u=730$.\\

Figure \ref{fig:value} shows the resulting value function and the charging policy for $t=0$, and $\nu_1=0.5$.
The optimal control is of threshold type with two threshold levels. For low energy prices, energy is bought at the maximum rate (red area). Then, if the energy price exceeds the lower threshold level, which is at a price of $s\approx25$ (depending on $q$), the optimal policy is to wait (green area). If the energy price further also exceeds the higher threshold level at $s\approx45$, then energy is sold (blue area).
For the value function we observe that $V$ grows in $q$. So the more energy is inside the storage facility, the more the storage facility is worth.
Furthermore, $V$ grows in $s$ for intermediate and high energy prices, but falls in $s$ for low energy prices, if $q$ is small, too. This is directly related to the resulting control policy. For low energy prices, energy is bought, and hence an increasing price while still in the area of buying decreases the value.

Furthermore, note that on extending the figure to negative values of $s$, the behaviour proceeds in the same way as for small values of $s$.\\

\begin{figure}[h]
\begin{center}
\includegraphics[scale=0.5]{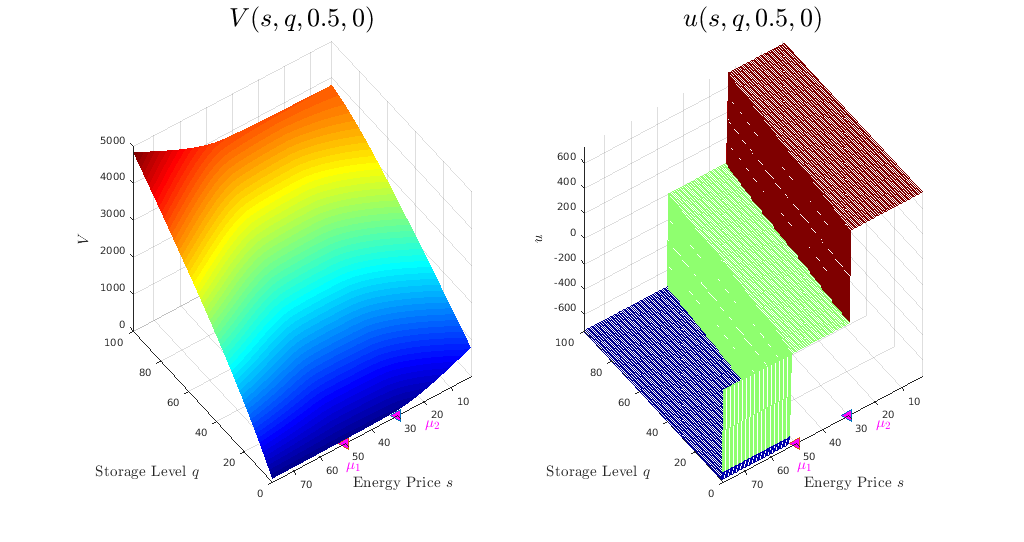}
\end{center}
\caption{The resulting value function (left) and charging policy (right).}\label{fig:value}
\end{figure}

Figure \ref{fig:entwicklung} shows the value function and the associated control for $t=0$ and for three different filter values, $\nu_1 \in \{0,0.5,1\}$. We observe that the higher $\nu_1$, the higher the value function. This has the following reason.
Since $\mu_1>\mu_2$, it means that if the probability for being in the state with the higher mean reversion level of the price process grows, then also the value of the energy storage facility grows.
For the optimal control higher $\nu_1$ implies that the lower threshold level, i.e., the one that separates the area of buying energy and waiting (between the red and the green areas) and also the higher threshold level, between the area of waiting and the area of selling energy (between the green and the blue areas) are shifted upwards, and hence in sum the area of buying (red area) grows and the area of selling (blue area) shrinks. This means for growing $\nu_1$ the storage manager will buy more energy since she expects a growing energy price. Then she could sell at higher prices later.\\

\begin{figure}[h]
\begin{center}
\includegraphics[scale=0.25]{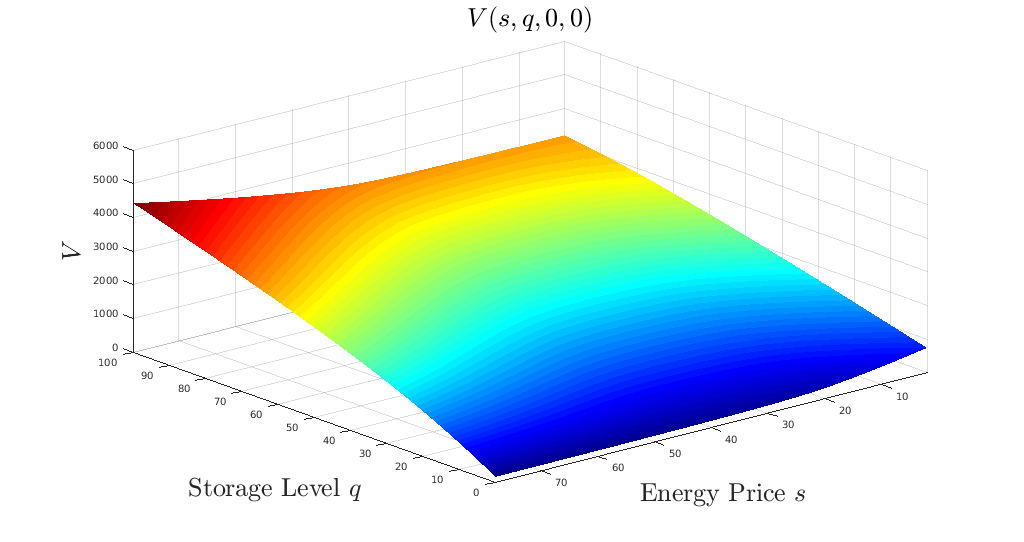} \includegraphics[scale=0.25]{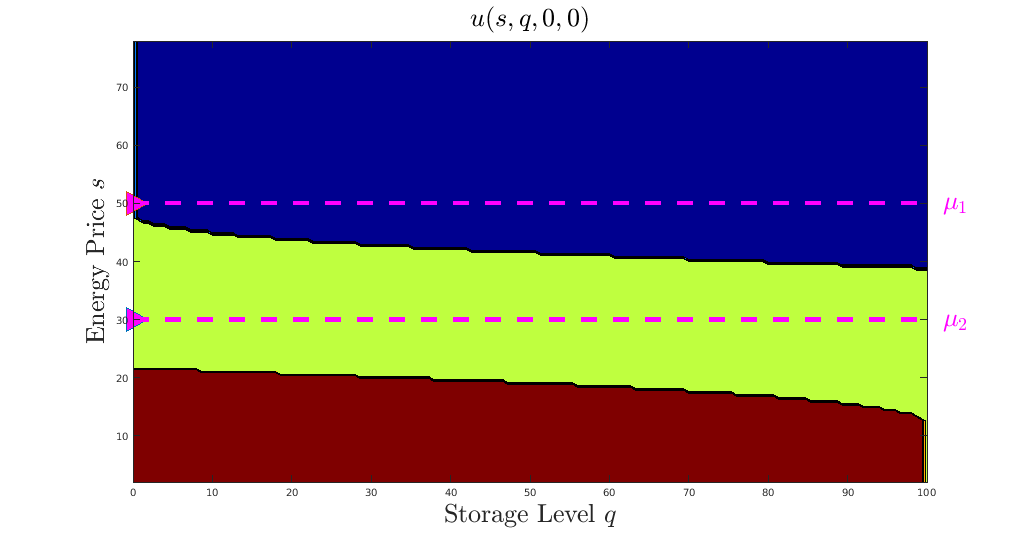}
\includegraphics[scale=0.25]{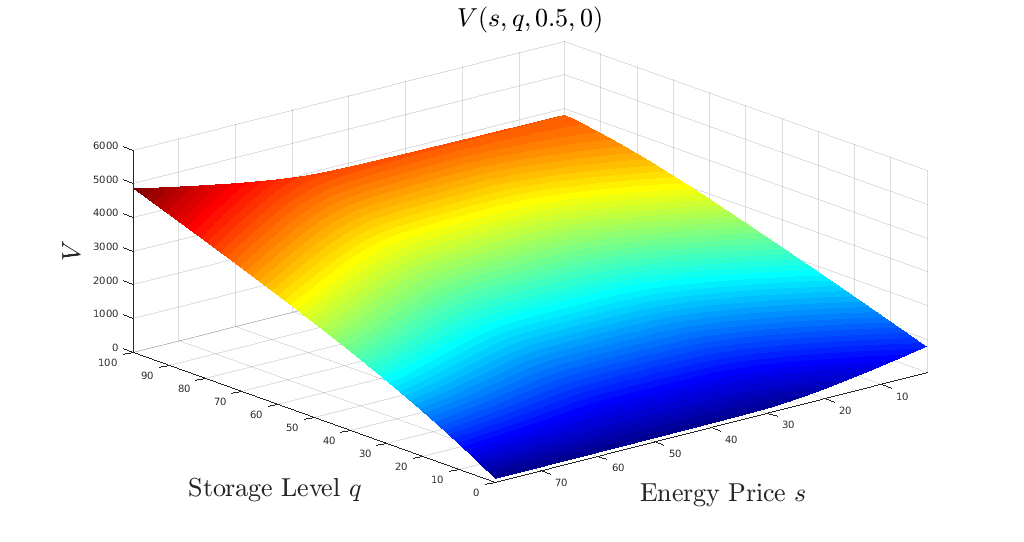} \includegraphics[scale=0.25]{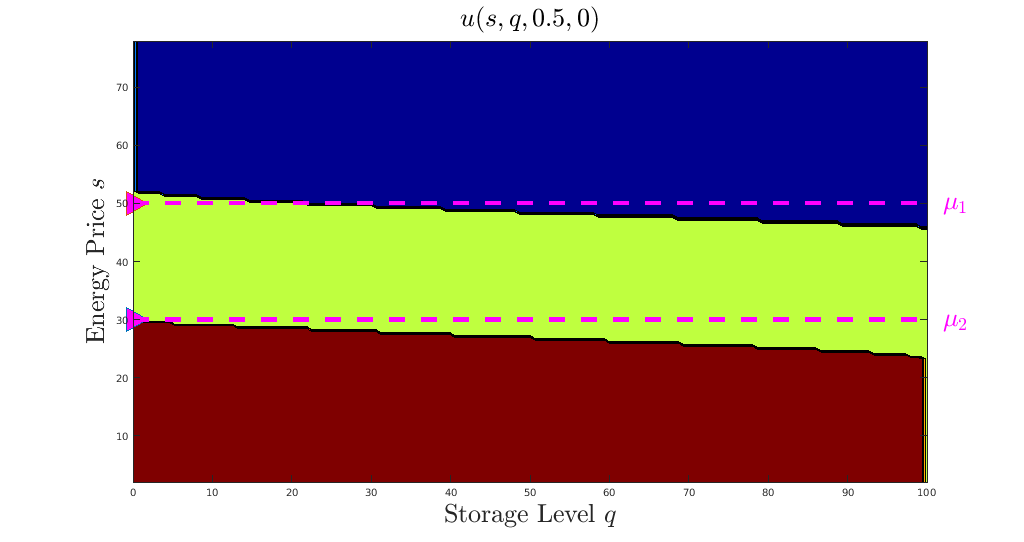}
\includegraphics[scale=0.25]{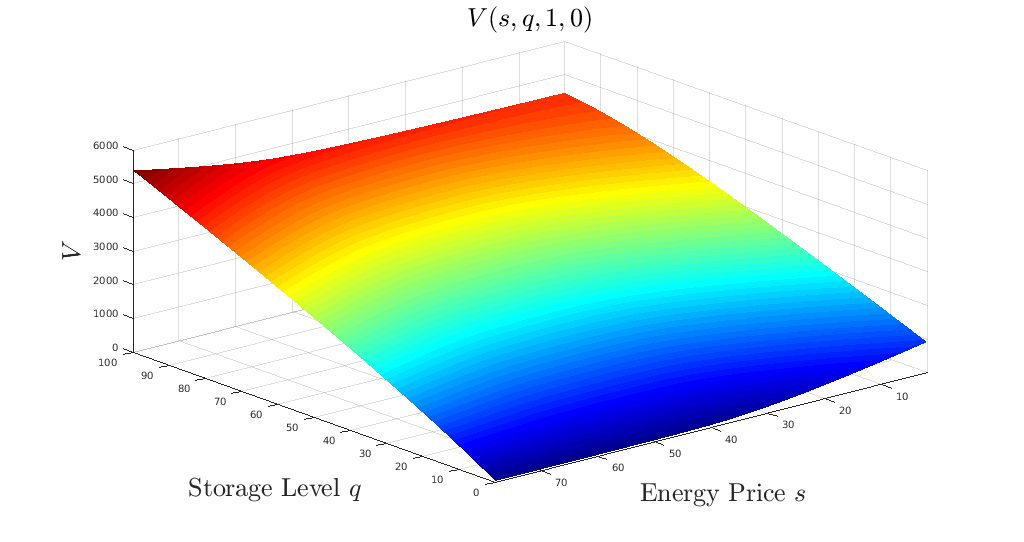} \includegraphics[scale=0.25]{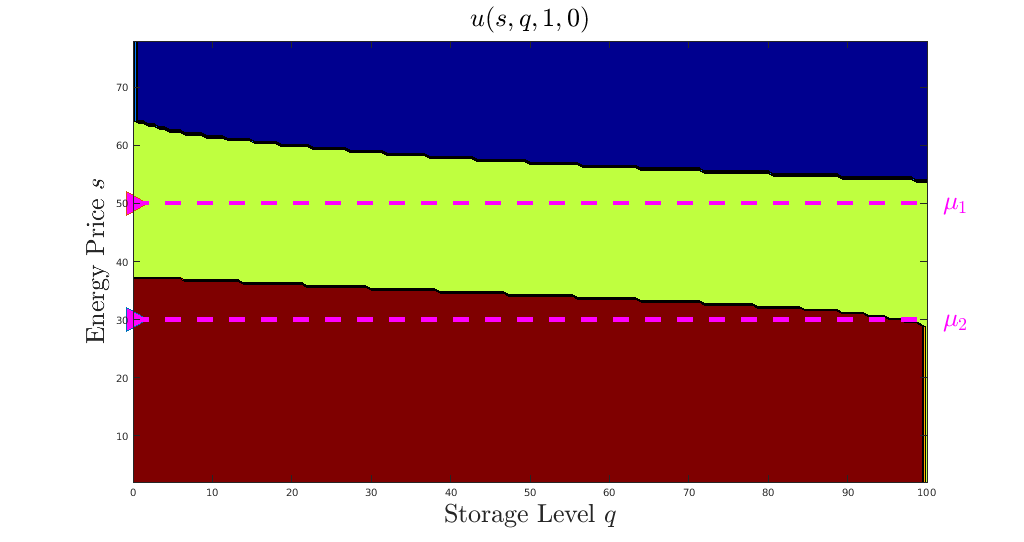}
\caption{The value functions (left) and charging policies (right) for growing values of $\nu_1$ (from up to down).}
\label{fig:entwicklung}
\end{center}
\end{figure}

In Figure \ref{fig:SPiQfix}, where we set $t=0$, and fix $q=50$, this effect becomes even more clear. While the value function grows gently in $\nu_1$, the levels between the buy- and wait region, and the wait- and sell region grow in $\nu_1$. So for higher $\nu_1$ the storage manager will buy more energy, since she expects to be able to sell it at higher prices later.

Note that growing $\nu_1$ does not necessarily imply a growing value in all parameter sets.
For example, consider the case of $\mu_1\gg \mu_2$. If the storage level is very low and the energy price is high (around $\mu_1$), then small $\nu_1$ implies falling prices such that in the future one can buy cheaper, whereas growing $\nu_1$ means that prices will probably stay high, making charging expensive and hence leading to a falling value. This effect is illustrated in Figure \ref{fig:fallval}, where we fix $s=\mu_1$ and compare the value function for $q=0$ (red) and $q=100$ (blue) in dependence of $\nu_1$.
Already for our parameter set, $V$ falls moderately in $\nu_1$ for $q=0$.\\

\begin{figure}[h]
\begin{center}
\includegraphics[scale=0.25]{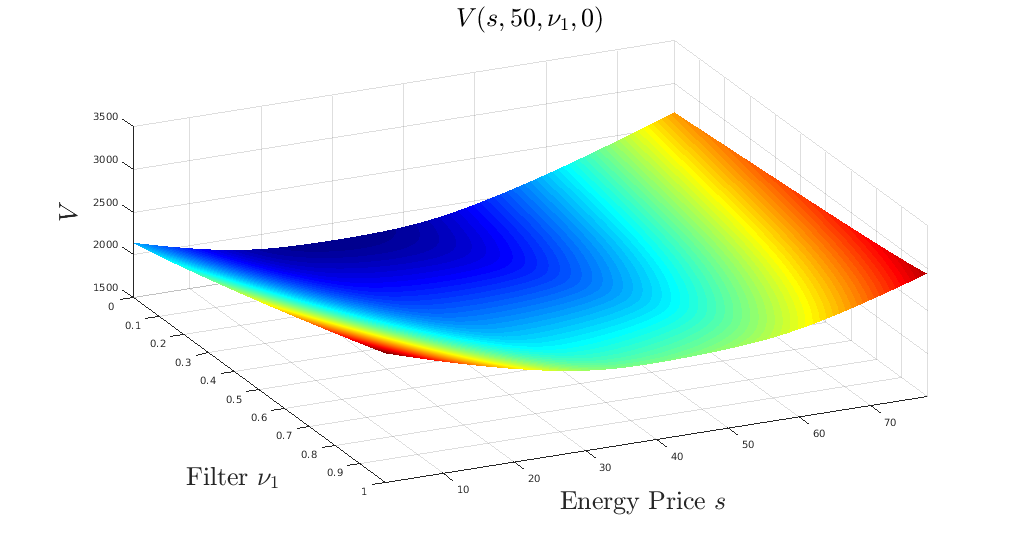}
\includegraphics[scale=0.25]{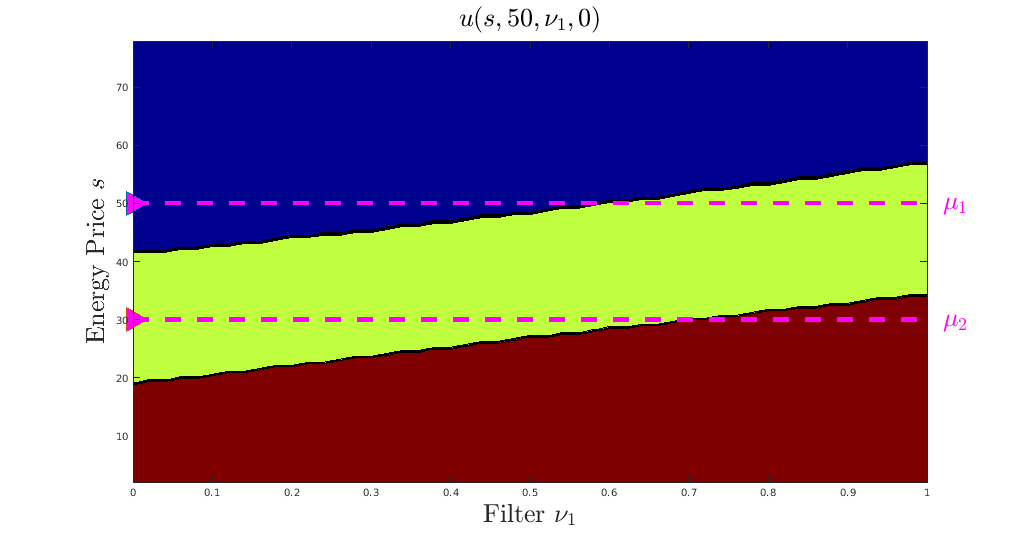}
\end{center}
\caption{The dependence of the value function (left) and the charging policy (right) on the filter.}
\label{fig:SPiQfix}
\end{figure}

\begin{figure}[h]
\begin{center}
\includegraphics[scale=0.25]{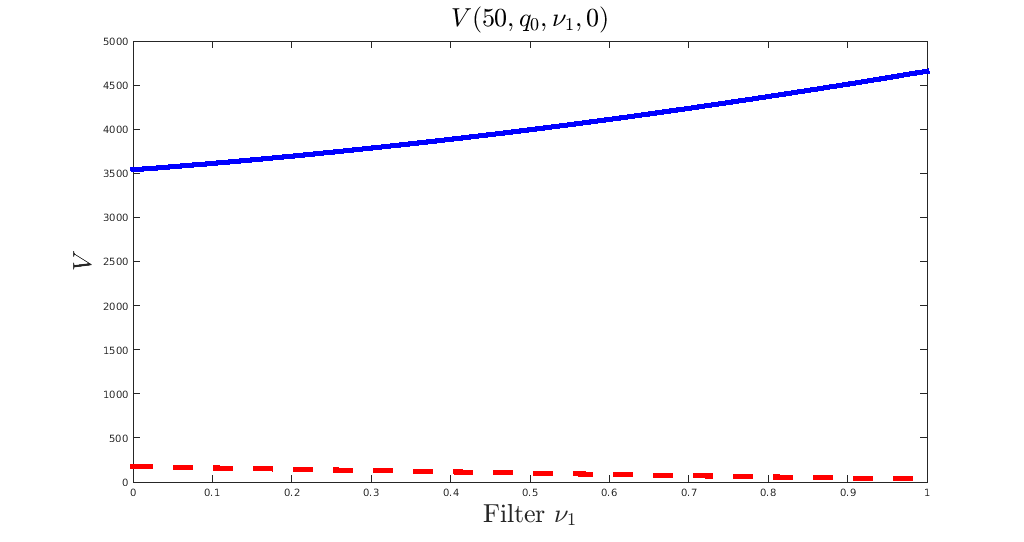}
\caption{The value function for $q_0=0$ (dashed red) and $q_0=100$ (blue).}
\label{fig:fallval}
\end{center}
\end{figure}

In Figure \ref{fig:TSQfix} we fix $q=50$, and $\nu_1=0.5$. We observe that the value function falls in $t$ since as $t$ grows, there is less time left to control the system and hence to earn money. Furthermore, at the end of the planning horizon, there is a multiplicative scrap rate $c_S$, which amplifies this effect. For the charging policy growing $t$ means that the levels between buying and waiting, and waiting and selling sink to avoid the penalty.\\

\begin{figure}[h]
\begin{center}
\includegraphics[scale=0.25]{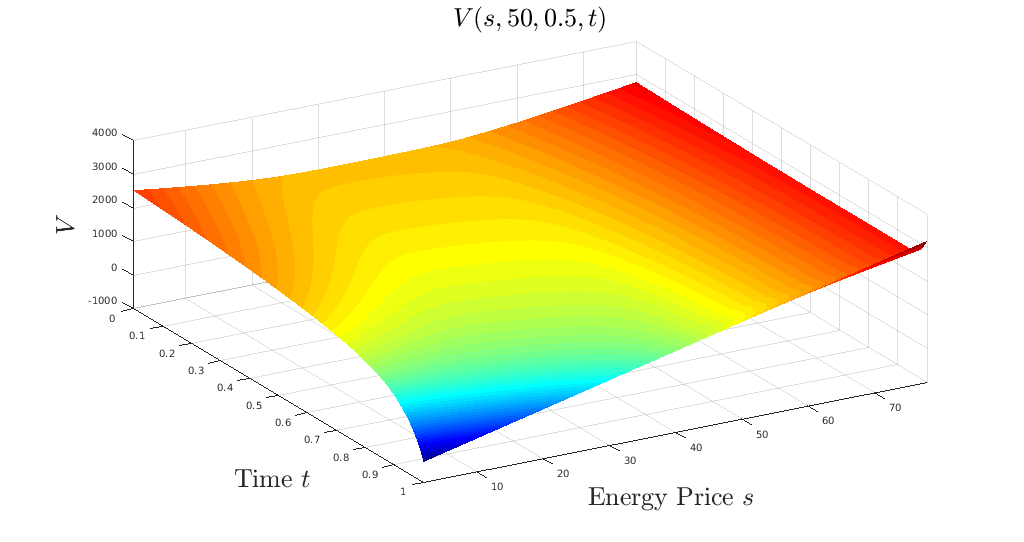}
\includegraphics[scale=0.25]{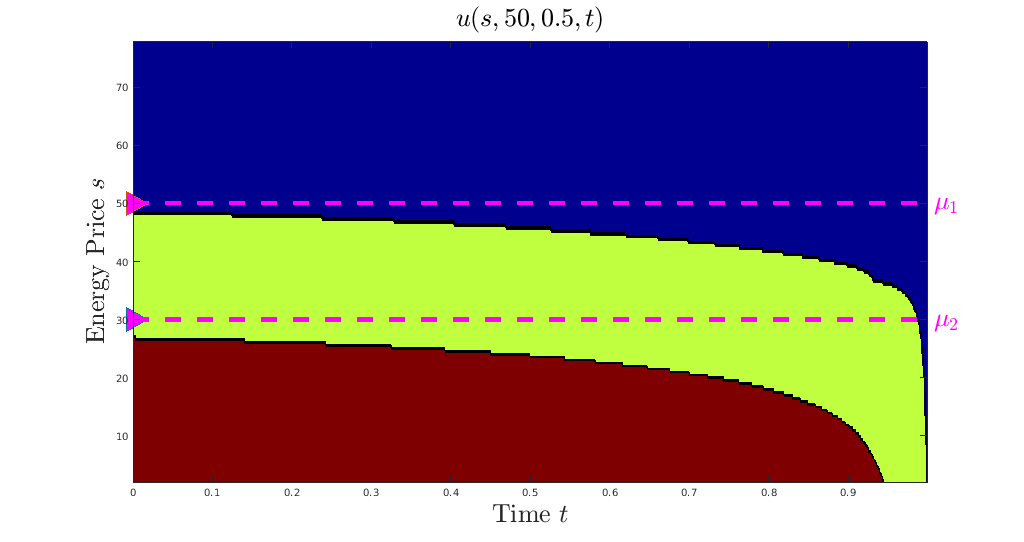}
\end{center}
\caption{The dependence of the value function (left) and the charging policy (right) on time.}
\label{fig:TSQfix}
\end{figure}

Finally, we change the planning horizon in our parameter set to $T=3$ and compare it to the case where $T=1$ for $t=0$, and $\nu_1=0.5$, see Figure \ref{fig:Stationar}. For longer planning horizon the value function is especially higher for small prices, because a longer time horizon increases the chance that prices grow. The effect becomes more clear on regarding the charging policy.
The lower threshold level is higher for $T=3$, so more energy will be stored to sell it at higher prices later.
Additionally, for high values of $q$ and growing planning horizon the higher threshold level is higher and hence the energy storage manager will start selling at higher prices.

For $T\to\infty$ value function and charging policy will converge to the stationary case, in which the value is autonomous.\\

\begin{figure}[h]
\begin{center}
\includegraphics[scale=0.25]{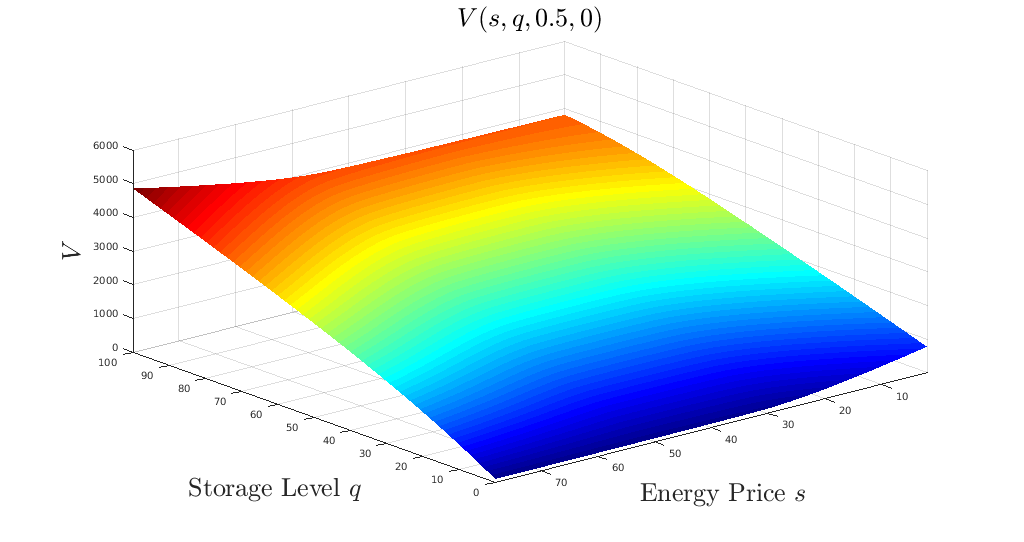} \includegraphics[scale=0.25]{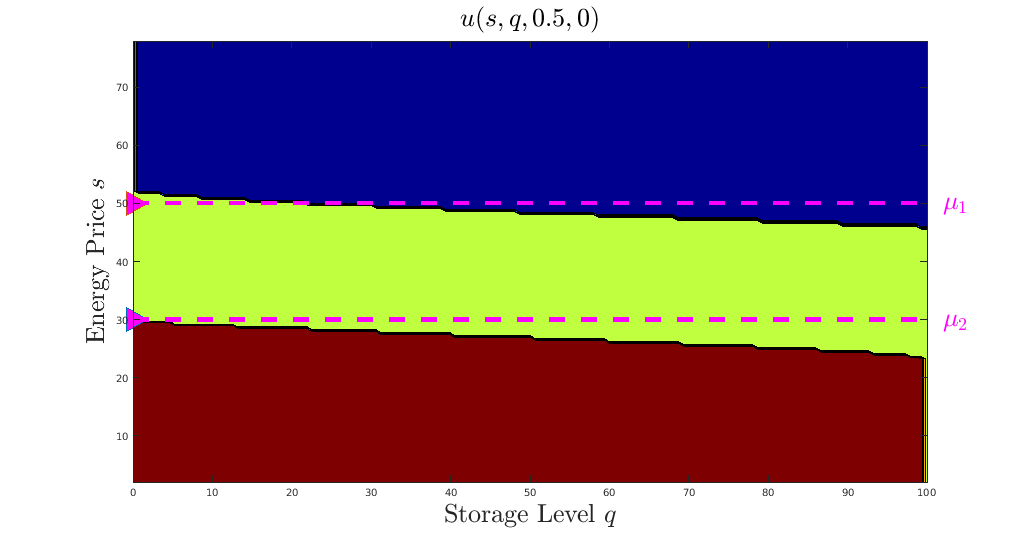}
\includegraphics[scale=0.25]{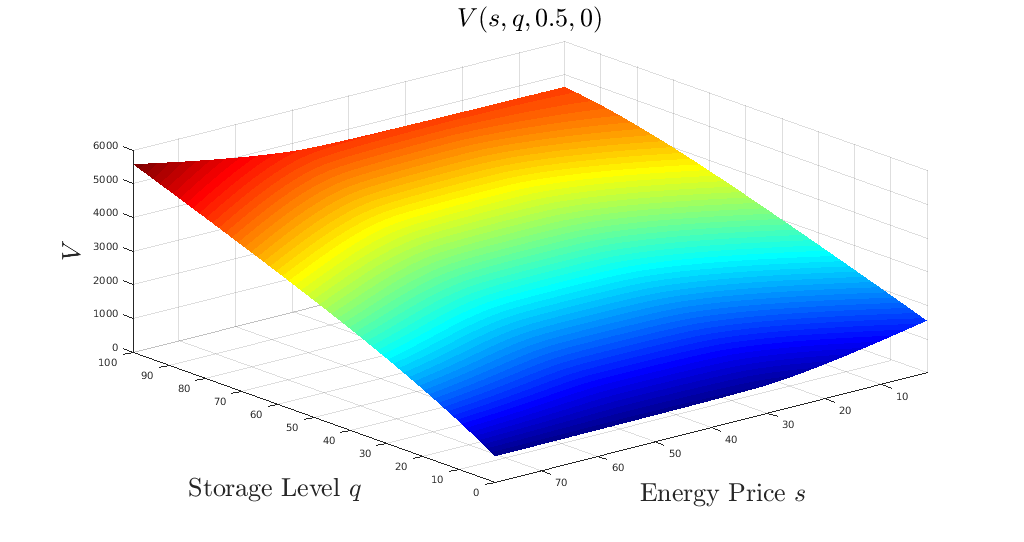} \includegraphics[scale=0.25]{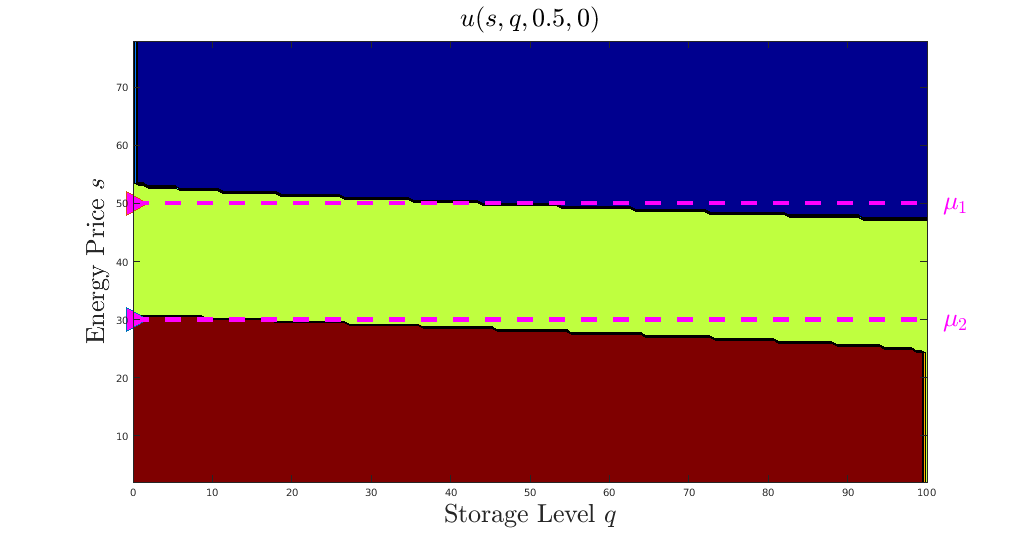}  
\caption{The value functions (left) and charging policies (right) for growing planning horizon ($T=1$ above, $T=3$ below).}
\label{fig:Stationar}
\end{center}
\end{figure}

In all our examples we observe that the candidate for the optimal control attains three states, namely the state of selling energy, the state of waiting, and the state of buying energy.
The functions $\underline s, \overline s$ from \eqref{eq:sws} describe the critical energy prices at which the optimal rate of charging switches between these states.
Hence, the resulting strategy is of threshold type.

\begin{remark}
 In \cite[Section 7]{sz12} it is shown that the value function is the unique viscosity solution of their HJB equation with fixed threshold strategy, if a condition implying concavity in a weak sense is satisfied.
A related result can also be obtained in our case.
Then, threshold strategies are optimal, if they lead to a smooth reward function, see Proposition \ref{prop:verification}.
\end{remark}

It remains to check whether threshold type strategies are admissible.
The question of admissibility of the resulting strategies reduces to the question of whether the underlying processes exist, or equivalently, whether the system of SDEs which describes the underlying processes has a solution.
It turns out that this issue is not straightforward.
Remember the underlying system

\begin{equation}
\begin{alignedat}{3}\label{eq:underlyingsystem}
dS_t&=\kappa \left(\sum\limits_{i=1}^D \mu_i \pi_t^i+K(t)-S_t \right ) \, dt +\sigma dB_t \, , &S_0&=s_0  \, ,\\
dQ_t&=\left(\underline u(Q_t) \1_{\{S_t \ge \underline s (S_t, Q_t,\pi_t,t)\}} + \overline u(Q_t) \1_{\{S_t \le \overline s (S_t,Q_t,\pi_t,t)\}}\right) dt \, ,\quad &Q_0&=q_0 \, , \\
d\pi_t&=\Lambda^{\top} \pi_t\, dt+ \sigma^{-1}\Bigl(\mathrm{diag}(\pi_t) a_t-\widehat{a}_t \pi_t \Bigr) d B_t \, ,  &\pi_0&=\nu_0\,.
\end{alignedat}
\end{equation}

The drift of system \eqref{eq:underlyingsystem} is discontinuous, hence we cannot apply the classical theory as mentioned in Section \ref{sec:Introduction}. Therefore, we study this issue in detail. In Section \ref{sec:Existence} we consider a more general setup.

\section{Existence and uniqueness result}
\label{sec:Existence}

We consider the following $d$-dimensional time-inhomogeneous SDE:

\begin{align}\label{eq:SDE1}
 d X_t=\alpha(X_t,t) dt + \beta(X_t,t) d W_t\,, \quad  X_0= x\,,
\end{align}
where $\alpha: \domain \longrightarrow \R^{d}$, $\beta: \domain\longrightarrow \R^{d \times d}$, and $W=(W_t)_{t\ge0}$ is a $d$-dimensional standard Brownian motion.\\

It is well-known that if $\alpha$ and $\beta$ are (locally) Lipschitz in $X_t$ and satisfy linear growth, then there exists a \uss{}.
For $\alpha$ only bounded and measurable and $\beta$ bounded, Lipschitz, and uniformly non-degenerate, i.e., 
there exists a constant $\lambda>0$ such that for all $x\in \R^d$ and all $v\in \R^d$ we
have $v^\top\beta(x)\beta(x)^\top v\ge \lambda v^\top v$,
\citet{zvonkin1974} and \citet{veretennikov1981} prove that there still exists a \uss{}.
A generalization is \citet{veretennikov1984}, where $\beta$ only needs to be non-degenerate in those components in which the drift is not Lipschitz.
Another generalization is \citet{zhang2005}, where only a locally integrable drift is required.
A first idea to handle a discontinuous drift would be to approximate $\alpha$ by smooth functions and analyze the approximation limit, see, e.g., \citet{krylov2002}.
A further existence and uniqueness result for the situation where the drift is only measurable is  proven in
\citet{meyer2010} with techniques from Malliavin calculus.
However, all of these results heavily depend on a non-degenerate diffusion coefficient.
For the degenerate situation \citet{halidias2006} prove an existence and uniqueness result for an SDE with an increasing drift coefficient by approximation with sub- and supersolutions.

A more recent result for the degenerate situation is \citet{sz14}, where the condition that the diffusion coefficient needs to be uniformly non-degenerate is replaced by a condition on the geometrical relation between the diffusion coefficient and the hypersurface along which the drift is discontinuous.
Their proof extends the idea from \citet{zvonkin1974}, however, it is neither required that the coefficients are bounded, nor that the drift is increasing.
\citet{sz14} consider a time-homogeneous SDE.
We extend the result from \citet{sz14} to the time-inhomogeneous case, in which conditions on the coefficients can be relaxed.
\citet{sz15} generalize the result from \citet{sz14} for the one-dimensional setting.
The most general result for time-homogeneous SDEs with discontinuous drift and possibly degenerate diffusion coefficient is \citet{sz2016b}, which is a generalization of \citet{sz15} to multidimensional SDEs.\\

Proving that the $d$-dimensional SDE \eqref{eq:SDE1} has a \uss{} is equivalent to proving that
\begin{align}\label{eq:SDE2}
 d\overline X_t= \overline \alpha(\overline X_t) dt +  \overline \beta(\overline X_t) dW_t\,, \quad \overline X_0=(x,0)\,,
\end{align}
has a \uss,
where $\overline \alpha: \domain \longrightarrow \R^{d+1}$, with $\overline \alpha_i = \alpha_i$ for $i=1,\ldots,d$, and $\overline \alpha_{d+1}\equiv1$, $\overline \beta: \domain \longrightarrow \R^{(d+1) \times d}$, with $\overline \beta_{ij} = \beta_{ij}$ for $i,j=1,\ldots,d$, and $\overline \beta_{(d+1),j}\equiv0$ for $j=1,\ldots,d$.\\
It seems a bit artificial to study system \eqref{eq:SDE2} instead of \eqref{eq:SDE1}, but we would like to apply results from \cite{sz14} and hence have to fit our problem into the setup therein.\\

\begin{assumption}\label{ass:f}
Consider a function $f:\domain\longrightarrow\R$ satisfying
 \begin{enumerate}
  \item[(f1)] $f \in C^{3,5,3}(\domnew)$;
  \item[(f2)] $\vert \frac{\partial f}{\partial x_1}\vert > 0$.
 \end{enumerate}
\end{assumption}

Under these assumptions there exists $e \in C^{3,5,3}(\domnew)$ such that
\begin{align*}
e(f(x,t),x_2,\ldots,x_d,t)&=x_1\,,\\
f(e(u,x_2,\ldots,x_d,t),x_2,\ldots,x_d,t)&=u\,.
\end{align*}

Later, we will allow $\alpha$ to be discontinuous along the hypersurface $\{(x,t) \in \domain:f(x,t)=0\}$.
For this, we define
\begin{equation}\label{eq:hatcoeff}
\begin{aligned}
\hat{\alpha}_1(u,x_2,\ldots,x_d,t) &=\frac{\partial f}{\partial t}+ \sum_{i=1}^d \alpha_i(u,x_2,\ldots,x_d,t) \frac{\partial f}{\partial x_i}+\frac{1}{2}\sum_{i,j=1}^d (\beta \beta^\top)_{ij} \frac{\partial^2 f}{\partial x_i \partial x_j}\,,\\
\hat{\alpha}_i(u,x_2,\ldots,x_d,t) &= \overline \alpha_i(u,x_2,\ldots,x_d,t)\,, \qquad i=2,\ldots,d+1\,,\\
\hat{\beta}_{1j} (u, x_2,\ldots,x_d,t) &= \sum_{i=1}^d  \beta_{ij} \frac{\partial f}{\partial x_i}\,, \qquad j=1,\ldots,d\,,\\
\hat{\beta}_{ij} (u, x_2,\ldots,x_d,t) &= \overline \beta_{ij}\,, \qquad i=2,\ldots,d+1\,;\,\,j=1,\ldots,d\,,
\end{aligned}
\end{equation}
and
\begin{align}\label{eq:SDE3}
d\hat X_t=\hat \alpha(\hat X_t)dt + \hat \beta(\hat X_t) dW_t\,, \qquad \hat X_0=(f(x,t),x_2,\ldots,x_d,t)\,,
\end{align}
where all missing arguments are $(e(u,x_2,\ldots,x_d,t),x_2,\ldots,x_d,t)$.\\

We make the following assumptions:
\begin{assumption}\label{ass:coefficients}
Let
 \begin{enumerate}
  \item[(c1)] $\beta \in C^{1,3,2}(\domnew)$;
  \item[(c2)] $\beta$ is such that $\hat \beta$ is Lipschitz;
  \item[(c3)] $\left\| \nabla f(x,t) \cdot \overline \beta(x,t)  \right\|^2\ge c > 0$ for some constant $c$ and for all $(x,t) \in \domain$;
  \item[(c4)] $\alpha$ is such that there exist functions $\alpha_+, \alpha_-\in C^{1,3,2}(\domnew)$ such that
\[
\alpha(u,x_2,\ldots,t)=\left\{
\begin{array}{ccc}
\alpha_+(x_1,x_2,\ldots,t)&\text{ if }&f(x,t)>0\\
\alpha_-(x_1,x_2,\ldots,t)&\text{ if }&f(x,t)<0;\\
\end{array}
\right.
\]
   \item[(c5)] $\alpha,\beta$ are such that $\hat\alpha,\hat\beta$ satisfy linear growth, i.e., there exist constants $c_1,c_2\ge0$ such that $\|\hat\alpha(x,t)\|+\|\hat\beta(x,t)\| \le c_1+c_2 \|x\|$.
 \end{enumerate}
\end{assumption}

\begin{remark}\label{rem:asscoeff}
The following conditions on $f$ and the coefficients $\alpha, \beta$ imply that items (c2) and (c5) in Assumptions \ref{ass:coefficients} hold. However, Assumptions \ref{ass:f} and \ref{ass:coefficients} contain less restrictive conditions on the coefficients \eqref{eq:hatcoeff}. 
 \begin{enumerate}
  \item[(f3)] $\frac{\partial f}{\partial x_i}$ is bounded and Lipschitz for $i=1,\ldots,d$, and there exist constants $c_1,c_2>0$ such that $\|\frac{\partial f}{\partial t}(x,t)\|\le c_1+c_2 \|x\|$;
  \item[(f4)] there exist constants $c_1,c_2>0$ such that $\|\frac{\partial f}{\partial x_i \partial x_j}(x,t)\|\le c_1+c_2 \|x\|$ for all $i,j=1,\ldots,d$;
  \item[(c2')] $\overline \beta$ is bounded and Lipschitz;
   \item[(c5')] $\overline \alpha$ satisfies linear growth, i.e., there exist constants $c_1,c_2\ge0$ such that $\|\alpha(x,t)\| \le c_1+c_2 \|x\|$.
 \end{enumerate}
\end{remark}

\begin{remark}
 Assumption \ref{ass:coefficients} (c3) is a crucial condition, which, in the case of a degenerate diffusion coefficient $\beta$, has a nice geometrical interpretation. It means that the diffusion must not be parallel to the hypersurface along which the drift is discontinuous. We refer to \cite{sz14} for more details.
\end{remark}

\begin{theorem}\label{thm:ex}
 Let Assumptions \ref{ass:f} and \ref{ass:coefficients} hold.
 Then \eqref{eq:SDE2} has a unique global strong solution.
\end{theorem}

For the proof we refer to Appendix \ref{app:proof}.


\section{Application to the energy storage process}
\label{sec:Energy}

Now we apply the result from Section \ref{sec:Existence} to prove admissibility of the resulting threshold strategies for the energy storage optimization problem studied in Section \ref{sec:opt}. We consider the case of $D=2$ as in the numerical study to allow for more intuition, but remark that the result from Section \ref{sec:Existence} can be applied to the general case as well. The question of admissibility of the resulting strategies reduces to the question of whether the underlying processes exist, or equivalently, whether the system of SDEs , which describes the underlying processes, has a solution. In the notation of Section \ref{sec:Existence}, using $\pi^2=1-\pi^1$, this system has the following form
\begin{align}\label{eq:systemnewnot}
dX_t&=\alpha(X_t,t) dt + \beta(X_t,t) d W_t \,,
\end{align} 
\begin{align*}
\alpha(X_t,t)&=\nonumber
\begin{pmatrix}
\kappa \Bigl(\mu_1\pi_t^1+\mu_2(1-\pi_t^1)+K(t)-S_t\Bigr) \\
\underline u(Q_t) \1_{\{S_t \ge \underline s (S_t, Q_t,\pi_t,t)\}} + \overline u(Q_t) \1_{\{S_t \le \overline s (S_t,Q_t,\pi_t,t)\}}   \\
\lambda_{11}\pi_t^1-\lambda_{22}(1-\pi_t^1)
\end{pmatrix}\, , \,\quad
\beta(X_t,t)= \nonumber
\begin{pmatrix}
\sigma & 0& 0 \\
0  & 0& 0 \\
\frac{\kappa}{\sigma}(\mu_1-\mu_2)\pi_t^1(1-\pi_t^1) & 0 & 0
\end{pmatrix}\, ,
\end{align*} 
$X=(S,Q,\pi^1)^\top$, $X_0=(s,q,\nu_1)^\top$, and $W=(B,W^2,W^3)^\top$ is a three-dimensional standard Brownian motion.\\

From the study in Section \ref{sec:opt} we get that the optimal control attains three states, namely the state of selling energy, the state of waiting, and the state of buying energy.
$\underline s, \overline s$ describe the critical energy prices at which the optimal rate of charging switches between these states.

For applying our existence and uniqueness result, we need the following relation:
\begin{align*}
S_t \ge \underline s(S_t,Q_t,\pi^1_t,t) \Leftrightarrow S_t \ge \underline b (Q_t,\pi^1_t,t)\,,\\
S_t \le \overline s(S_t,Q_t,\pi^1_t,t) \Leftrightarrow S_t \le \overline b (Q_t,\pi^1_t,t)\,,
\end{align*}
where $\underline b, \overline b$ are the resulting levels where the strategy switches between the states, i.e., the switching levels are actually functions of $(q,\nu_1,t)$.
In Figure \ref{fig:barriers}, $\underline b$ is the boundary between the green and the blue areas, and $\overline b$ is the boundary between the green and the red areas.
\begin{figure}[h]
\begin{center}
\includegraphics[scale=0.3]{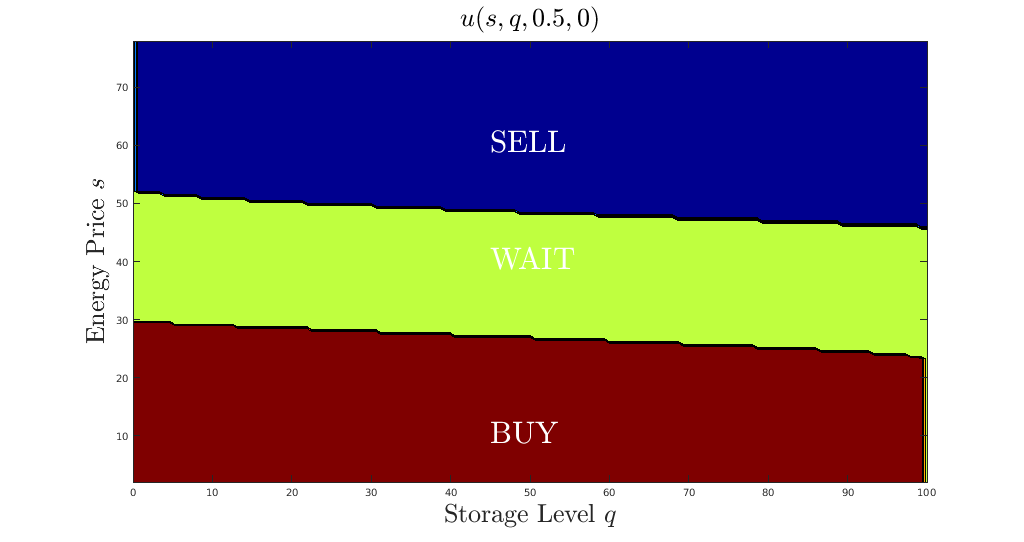}
\end{center}
\caption{The switching levels $\underline b$ and $\overline b$.}\label{fig:barriers}
\end{figure}

Theoretically, the existence of the functions $\underline b, \overline b$ can be proven by applying the implicit function theorem, where the crucial conditions are $  \frac{\partial \underline s}{\partial s} -1 \ne 0, \frac{\partial \overline s}{\partial s} -1 \ne 0$, or equivalently  $  V_{s q} \ne 1$ for $(s,q,\nu_1,t) \in \R \times [\underline{q},\overline{q}] \times[0,1]\times [0,T]$, and $\Phi_{sq}=c_S\ne1$ for $t=T$ in a neighbourhood of the switching levels. Figure \ref{fig:Vmixed} shows the mixed derivative of the resulting value function $V$. In our numerical examples the condition was always satisfied globally.
\begin{figure}[h]
\begin{center}
\includegraphics[scale=0.35]{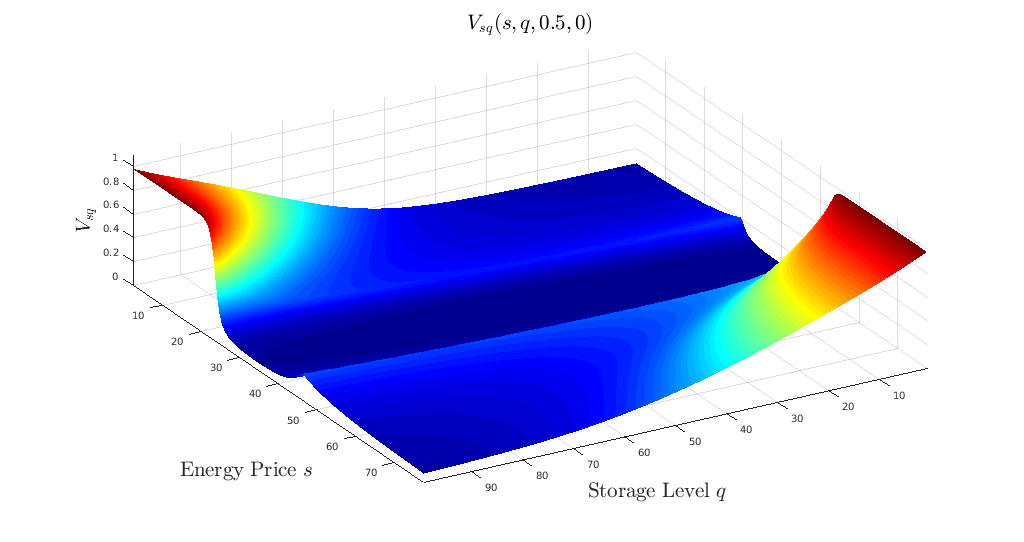}
\end{center}
\caption{The mixed derivative $V_{s q}$.}\label{fig:Vmixed}
\end{figure}

\paragraph{Existence and uniqueness proof.}

Now we prove the existence of a unique solution to system \eqref{eq:systemnewnot} using the result from Section \ref{sec:Existence}.\\

First of all, we define the functions $\underline f(s,q,\nu_1,t)=s-\underline b(q,\nu_1,t)$ and $\overline f(s,q,\nu_1,t)=s-\overline b(q,\nu_1,t)$.
With this, we can rewrite system \eqref{eq:systemnewnot} such that
\begin{align}\label{eq:systemnewnot2}
\alpha(X_t,t)&=
\begin{pmatrix}
\kappa \Bigl(\mu_1\pi_t^1+\mu_2(1-\pi_t^1)+K(t)-S_t\Bigr) \\
\underline u(Q_t) \1_{\{f(S_t,Q_t,\pi^1_t,t)\ge0\}} + \overline u(Q_t) \1_{\{f(S_t,Q_t,\pi^1_t,t)\le0\}}   \\
\lambda_{11}\pi_t^1-\lambda_{22}(1-\pi_t^1)
\end{pmatrix}\,.
\end{align}
Now we are ready to check whether Assumptions \ref{ass:f} and \ref{ass:coefficients} are fulfilled.
We begin with checking whether $\alpha$ is of class $C^{1,3,2}(\R \times ([\underline{q},\overline{q}]\times[0,1]) \times [0,T])$ outside the discontinuities. This holds since by assumption $K \in {C}^2([0,T])$, and since the functions $\underline u, \overline u$ are sufficiently smooth. Furthermore, $\beta_{i,j} \in C^{1,3,2}(\R \times( [\underline{q},\overline{q}]\times[0,1])\times [0,T])$ for $i,j=1,2,3$.
For the next condition we have to check whether $\underline f, \overline f \in C^{3,5,3}(\R \times( [\underline{q},\overline{q}]\times[0,1] )\times [0,T]  )$. Clearly, $\frac{\partial^3 \underline f}{\partial s^3} = \frac{\partial^3 \overline f}{\partial s^3} =0$.
However, for this condition to be satisfied, we also need that $\underline b, \overline b \in C^{5,3}(([\underline{q},\overline{q}]\times[0,1]) \times [0,T])$.

Furthermore, $\left|\frac{\partial \underline f}{\partial s} \right|, \left|\frac{\partial \overline f}{\partial s} \right|=1$.
It is crucial that the non-parallelity condition holds for both $(f,b)\in\{(\underline f, \underline b),(\overline f,\overline b)\}$:
\begin{align*} &\| \nabla  f(s,q,\nu_1,t) \cdot \overline \beta(s,q,\nu_1,t) \|^2 \\
 &=\|(1,- b_{q}(q,\nu_1,t),- b_{\nu_1}(q,\nu_1,t),- b_{t}(q,\nu_1,t)) \cdot 
\begin{pmatrix}
\sigma & 0& 0& 0\\ 0 & 0 & 0& 0\\ \frac{\kappa}{\sigma}(\mu_1-\mu_2)\nu_1(1-\nu_1)  & 0& 0& 0\\ 0& 0& 0& 0
\end{pmatrix}
\|^2 \\ &=| \sigma-  \frac{\kappa}{\sigma}(\mu_1-\mu_2)\nu_1(1-\nu_1) b_{\nu_1}(q,\nu_1,t)|\,,
\end{align*}
is positive, if $b_{\nu_1}(q,\nu_1,t) \neq \frac{\sigma^2}{\kappa (\mu_1-\mu_2)\nu_1(1-\nu_1) }$ for all $(q,\nu_1,t) \in [\underline q, \overline q]\times (0,1)\times[0,T]$,
where $\overline \beta$ is an extension of $\beta$ by a zero line and a zero row.
Note that we checked this condition in all our numerical examples and it was always satisfied.\\

Further we need $\hat \beta$ to be Lipschitz,
and that the linear growth condition holds for $\hat \alpha$ and $\hat \beta$, which is both fulfilled, if $\underline b, \overline b$ are sufficiently smooth functions on a bounded domain and since $K$ satisfies linear growth by assumption.

\begin{proposition}
If  for both $b\in\{\underline b, \overline b\}$, $b \in C^{5,3}(([\underline{q},\overline{q}] \times[0,1])\times [0,T])$ and $b_{\nu_1}(q,\nu_1,t) \neq \frac{\sigma^2}{\kappa (\mu_1-\mu_2)\nu_1(1-\nu_1) }$ for all $(s,q,\nu_1,t) \in \R\times[\underline q, \overline q]\times (0,1)\times[0,T]$, then system \eqref{eq:systemnewnot} has a unique global strong solution.
\end{proposition}

\begin{proof}
 It remains to extend the result from Section \ref{sec:Existence} to the situation with two discontinuities. As $\underline b, \overline b$ are distinct, i.e., they do not cross or converge to each other in $[\underline q,\overline q]\times [0,1] \times [0,T]$, we can find a hypersurface $\gamma \in C^{\infty}$ which lies between $\overline b, \underline b$. Furthermore, due to Theorem \ref{thm:ex}, a unique global strong solution to system \eqref{eq:systemnewnot} exists in the area below $\gamma$ as well as in the area above $\gamma$. In $\gamma$ the coefficients are the same for the limit from above and from below. Therefore, due to the strong Markov property, the solution from above and from below may be pasted together in $\gamma$.
\end{proof}

Thus, the resulting strategies are indeed admissible.\\

Note that on studying practical observations we found that the optimal value function is relatively invariant w.r.t. small changes in $\underline b, \overline b$. Therefore, we are optimistic that we can find $C^{5,3}(([\underline{q},\overline{q}] \times[0,1])\times [0,T])$ functions which are sufficiently close to the true -- possibly less smooth -- switching barriers such that the strategy switching at these $C^{5,3}(([\underline{q},\overline{q}] \times[0,1])\times [0,T])$ functions produces a reward function that is very close to the optimal one.

\section{Outlook}
\label{sec:Outlook}

On the one hand we have studied an optimization problem for an energy storage facility under partial information.
On the other hand we have shown how the issue of existence of uniqueness of processes controlled by policies of threshold type can be handled in general.
We emphasize that such kinds of SDEs appear in a wide range of applied mathematics. We will list some further fields in which our result or the results presented in \cite{sz15, sz2016b, sz14} are applicable.\\

A question arising in inventory management is the optimal stock holding problem, see, e.g., \cite{tempelmeier2008}.
There the stockpile $L=(L_t)_{t\ge0}$ is modelled as
\begin{align*}
dL_t=(u_t - D_t)dt\,, \qquad L_0=\ell\,, \quad t \in [0,T] \, ,
\end{align*}
where $D=(D_t)_{t\ge0}$ is the demand for the good in the stock, which can be modelled as a diffusion process independent of $L$. $u_t \in [0,M_u]$ is the amount of goods ordered for the stock at time $t$, bounded by $M_u$, which serves as the control. For the resulting threshold type controls, our result is necessary to prove existence and uniqueness of the two-dimensional process $(L,D)$.\\

In \cite{benth2007} the valuation of swing options on electricity markets is studied.
There, a similar SDE as in Section \ref{sec:Energy} appears:
\begin{align*}
dZ_t=u_t dt\,, \qquad Z_0=0 \,, \quad t \in [0,T] \, ,
\end{align*}
where $Z_t$ is the amount of energy available at time $t$. $u_t \in [0,M_u]$ is again the control, bounded by $M_u$, which describes the amount of energy that is bought in case the option is exercised, and which in addition to $Z$ also depends on the energy price, cf.~Section \ref{sec:Energy}. In this setup, the control is again of threshold type with one threshold level. Thus, again our result is potentially applicable there.\\

Overall we saw that threshold strategies appear in a wide range of fields of applied mathematics and hence the result presented herein potentially solves a large variety of problems.


\appendix
\section{Proof of the existence and uniqueness result}
\label{app:proof}

In principle we would like to apply the results from \cite{sz14}, where only the time-homogeneous case is considered, to prove existence and uniqueness in the time-inhomogeneous case. Herein, the conditions on the time dimension are less restrictive. To prove Theorem \ref{thm:ex} we need to adapt the techniques presented in \cite{sz14} to our situation. For this, we prove two lemmas.\\

First of all, we assume $\alpha$ to be discontinuous only along $\{x_1=0\}$, i.e., the special case $f(x,t)=x_1$.
For this case \eqref{eq:SDE2} and \eqref{eq:SDE3} coincide.
Then we seek to remove the discontinuity from the drift in a way such that the remaining coefficients are locally Lipschitz.

For this we define a transformation $Z=G(\overline X)$ with $G:\domain\longrightarrow\domain$ by
\begin{align*}
G(x_1,\ldots,x_d,t)&:=(g_1(x),x_2+g_2(x),\ldots,x_d+g_d(x),t)\,,
\end{align*}
and choose $g_1,\ldots,g_{d}$ in a way such that $\frac{\partial G_k}{\partial x_i}(0,x_2,\ldots,x_d,t)=\delta_{i,k}$ for $i,k=1,\ldots,d+1$. Then $G$ is locally invertible with local inverse $H$.
Thus, we can now consider the transformed SDE
\begin{align}\label{eq:SDEtrans}
dZ_t=\left(\nabla G(H(Z_t))\overline \alpha(H(Z_t)) 
+ \tr\left(\overline\beta(H(Z_t))^\top \nabla^2 G(H(Z_t))\overline\beta(H(Z_t))\right)\right)dt
+\nabla G(H(Z_t))\overline\beta(H(Z_t))dW_t\,.
\end{align}

\begin{lemma}\label{lem:lip}
Let Assumption \ref{ass:coefficients} hold.
Then there exist functions $g_1,\ldots,g_{d}:\domain \longrightarrow \R$ such that the coefficients of \eqref{eq:SDEtrans} are locally Lipschitz, and such that $G$ is locally invertible around $\{x_1=0\}$.
\end{lemma}

The following proof is constructive, but note that the presented construction is not unique.

\begin{proof}[Proof of Lemma \ref{lem:lip}]
The construction of $g_1,\ldots,g_d$ is done like in \cite{sz14}.
For $g_1$ consider
\begin{equation}\label{eq:help1}
\begin{aligned}
dZ^1_t
&= \left[\alpha_1 \frac{\partial g_1}{\partial x_1} + \frac{1}{2} (\beta \beta^\top)_{11}
 \frac{\partial^2 g_1}{\partial x_1^2} \right] d t+  \sum_{i=2}^d \alpha_i \frac{\partial g_1}{\partial x_i} dt
 +\frac{\partial g_1}{\partial t} dt 
 \\ &+ \frac{1}{2} \sum_{i,j=1}^d(1-\delta_{11}(i,j)) (\beta \beta^\top)_{ij} \frac{\partial^2 g_1}{\partial x_i \partial x_j} dt +  \sum_{i,j=1}^d  \beta_{ij} \frac{\partial g_1}{\partial x_i} dW_t^j\,.
\end{aligned}
\end{equation}
Now choose
\begin{align*}
g_1(x,t) = \int_0^{x_1} \exp \left(- \int_0^\xi \frac{2 \alpha_1(s,x_2,\dots,x_d,t)}{(\beta \beta^\top)_{11}(s,x_2,\dots,x_d,t)}\, ds \right) d\xi \,,
\end{align*}
which ensures that $\alpha_1 \frac{\partial g_1}{\partial x_1} + \frac{1}{2} (\beta \beta^\top)_{11}
 \frac{\partial^2 g_1}{\partial x_1^2}$ vanishes globally and that for $k>1$, $\alpha_k \frac{\partial g_1}{\partial x_k}$ are locally Lipschitz since $\alpha_k$ is locally bounded and $\frac{\partial g_1}{\partial x_k}$ vanishes on $\{x_1=0\}$.
 Furthermore, $\frac{\partial g_1}{\partial t}$ is locally Lipschitz, since $g_1$ is $C^2$ in the time dimension due to Assumption \ref{ass:coefficients}.
 The terms in the second line of \eqref{eq:help1} are locally Lipschitz, which follows from Assumption \ref{ass:coefficients} together with \cite[Lemma 2.7]{sz14}.\\

For $k>1$ consider
\begin{align*}
dZ^k_t
&= \left[\alpha_k+\alpha_1 \frac{\partial g_k}{\partial x_1} + \frac{1}{2} (\beta \beta^\top)_{11}
 \frac{\partial^2 g_k}{\partial x_1^2} \right] d t+  \sum_{{i=2}}^d \alpha_i \frac{\partial g_k}{\partial x_i} dt
 \\ &+\frac{\partial g_k}{\partial t} dt + \frac{1}{2} \sum_{i,j=1}^d(1-\delta_{11}(i,j)) (\beta \beta^\top)_{ij} \frac{\partial^2 g_k}{\partial x_i \partial x_j} dt +  \sum_{i,j=1}^d  \beta_{ij}\left (1+\frac{\partial g_k}{\partial x_i}\right) dW_t^j\,,
\end{align*}
where we choose
\begin{align*}
g_k(x_1,\ldots,x_d,t) := \int_0^{x_1} C_k(\xi,x_2,\ldots,x_d,t)\exp \left(- \int_0^\xi
\frac{2\alpha_1(s,x_2,\ldots,x_d,t)}{(\beta \beta^\top)_{11}(s,x_2,\ldots,x_d,t)} \,ds \right) d\xi \,,
\end{align*}
with
\begin{align*}
C_k(\xi,x_2,\ldots,x_d,t):=
-\int_0^{\xi} \frac{2\alpha_k(\eta,x_2,\ldots,x_d,t)}{(\beta
\beta^\top)_{11}(\eta,x_2,\ldots,x_d,t)}\exp \left(\int_0^\eta
\frac{2\alpha_1(s,x_2,\ldots,x_d,t)}{(\beta
\beta^\top)_{11}(s,x_2,\ldots,x_d,t)} \,ds \right) d\eta\,.
\end{align*}
This guarantees that also the drift of $Z^k$ is locally Lipschitz.\\

In the time-direction we set $dZ_t^{d+1}=dt$.
This closes the proof.
\end{proof}

Since $H \notin C^{2,1}(\R^d \times \R_0^+)$ we need to show that It\^o's formula nevertheless holds for a class of functions including $H$.

\begin{lemma}\label{lem:ito}
Let $\gamma:\cD\subseteq \domain\longrightarrow\R$, $\gamma \in C^1(\cD)$, $\frac{\partial^2 \gamma}{\partial x_i \partial x_j}(x,t) \in C(\cD)$ for all $i=1,\dots,d, j=2,\dots,d$, $\frac{\partial^2 \gamma}{\partial x_1^2}(x,t) \in C((\domnew) \cap \cD)$, $\sup_{\{x \vert x_1\ne 0\}}|\frac{\partial^2 \gamma}{\partial x_1^2}(x,t)|<\infty$.
For an It\^o process $Y=(Y_t)_{t\ge 0}$, let $\zeta=\inf \{ t>0 \, \vert \, Y_t \notin \cD \}$.

Then for all $t \ge 0$
 \begin{align*}
\gamma(Y_t,t)&=\gamma(Y_0,0)+\frac{\partial \gamma}{\partial s}(Y_s,s) \,ds+\sum_{i=1}^d \int_0^{t\wedge \zeta} \frac{\partial \gamma}{\partial y_i}(Y_s,s)\,dY_s^i
+\frac{1}{2} \sum_{i,j=1}^d \int_0^{t\wedge \zeta} \frac{\partial^2 \gamma}{\partial y_i\partial y_j}(Y_s)\,d[Y^i,Y^j]_s\,.
\end{align*}

\end{lemma}

\begin{proof}
The proof runs along the same lines as in \cite{sz14}.
Additionally, we use the fact that for the regular It\^o formula,
we only require $\gamma \in C^1$ in the time-direction.
\end{proof}

Now, we are ready to prove the main theorem of Section \ref{sec:Existence}.

\begin{proof}[Proof of Theorem \ref{thm:ex}]
 We start with the case where $\alpha$ is discontinuous along $\{x_1=0\}$.
 Due to Lemma \ref{lem:lip} the coefficients of \eqref{eq:SDEtrans} are locally Lipschitz,
 which guarantees that the transformed SDE has a \usol.\\
 Furthermore, due to Lemma \ref{lem:ito}, It\^o's formula holds for the inverse $H$ of the transformation.
 By setting $\overline X=H(Z)$ and applying It\^o's formula to $H$, we get that \eqref{eq:SDE2} has a \usol.\\
 Now, we can directly apply \cite[Theorem 3.2]{sz14} which, by concatenating local solutions,
 guarantees existence and uniqueness of a unique maximal local solution to \eqref{eq:SDE2}.\\
 Additionally, we can directly apply \cite[Theorem 3.3]{sz14} to obtain that explosion does not happen in finite time.
 Hence, \eqref{eq:SDE2} has a \uss{} in case $\alpha$ is discontinuous along $\{x_1=0\}$.\\
 For the general case, where $\alpha$ is discontinuous along $\{(x,t)\in \domain:f(x,t)=0\}$,
 we have that due to Assumptions \ref{ass:f} and \ref{ass:coefficients}, system \eqref{eq:SDE3}
 has a \uss{} $\hat X$.
 Setting $X=e(\hat X)$ and applying the classical It\^o formula, we get that \eqref{eq:SDE2}
 has a \uss.
\end{proof}


\section*{Acknowledgements}

The authors thank R\"udiger Frey (WU Vienna), Gunther Leobacher (Johannes Kepler University Linz), and Ralf Wunderlich (BTU Cottbus-Senftenberg) for valuable discussions that improved this paper.\\
M. Sz\"olgyenyi is supported by the Vienna Science and Technology Fund (WWTF): Project MA14-031.\\
A part of this paper was written while M. Sz\"olgyenyi was member of the Department of Financial Mathematics and Applied Number Theory, Johannes Kepler University Linz, 4040 Linz, Austria.\\
During this time, M. Sz\"olgyenyi was supported by the Austrian Science Fund (FWF): Project F5508-N26, which is part of the Special Research Program "Quasi-Monte Carlo Methods: Theory and Applications".



\end{document}